\documentclass[a4paper,onecolumn,oneside,10pt]{article}
\usepackage{geometry}
\usepackage{graphics}
\usepackage{epstopdf}
\usepackage[dvips]{epsfig}

\input xy
\xyoption{all}
\usepackage{comment}
\usepackage{hyperref}

\usepackage[T1]{fontenc}
\usepackage{amsthm}

\usepackage{mathtools}
\usepackage{mathdots,enumerate}
\usepackage{textcomp}
\usepackage{siunitx}
\usepackage{amssymb,amsfonts}
\usepackage{amsmath,paralist,enumitem,mathrsfs,graphicx}
\usepackage{epstopdf}
\usepackage{pgfplots}
\newcommand\norm[1]{\left\lVert#1\right\rVert}

\usetikzlibrary{pgfplots.dateplot}
\renewcommand{\qedsymbol}{$\blacksquare$}
\newcommand\restr[2]{{% we make the whole thing an ordinary symbol
		\left.\kern-\nulldelimiterspace % automatically resize the bar with \right
		#1 % the function
		\vphantom{\big|} % pretend it's a little taller at normal size
		\right|_{#2} % this is the delimiter
}}

\makeatletter
\let\save@mathaccent\mathaccent
\newcommand*\if@single[3]{%
	\setbox0\hbox{${\mathaccent"0362{#1}}^H$}%
	\setbox2\hbox{${\mathaccent"0362{\kern0pt#1}}^H$}%
	\ifdim\ht0=\ht2 #3\else #2\fi
}
\newcommand*\rel@kern[1]{\kern#1\dimexpr\macc@kerna}
\newcommand*\widebar[1]{\@ifnextchar^{{\wide@bar{#1}{0}}}{\wide@bar{#1}{1}}}
\newcommand*\wide@bar[2]{\if@single{#1}{\wide@bar@{#1}{#2}{1}}{\wide@bar@{#1}{#2}{2}}}
\newcommand*\wide@bar@[3]{%
	\begingroup
	\def\mathaccent##1##2{%
		\let\mathaccent\save@mathaccent
		\if#32 \let\macc@nucleus\first@char \fi
		\setbox\z@\hbox{$\macc@style{\macc@nucleus}_{}$}%
		\setbox\tw@\hbox{$\macc@style{\macc@nucleus}{}_{}$}%
		\dimen@\wd\tw@
		\advance\dimen@-\wd\z@
		\divide\dimen@ 3
		\@tempdima\wd\tw@
		\advance\@tempdima-\scriptspace
		\divide\@tempdima 10
		\advance\dimen@-\@tempdima
		\ifdim\dimen@>\z@ \dimen@0pt\fi
		\rel@kern{0.6}\kern-\dimen@
		\if#31
		\overline{\rel@kern{-0.6}\kern\dimen@\macc@nucleus\rel@kern{0.4}\kern\dimen@}%
		\advance\dimen@0.4\dimexpr\macc@kerna
		\let\final@kern#2%
		\ifdim\dimen@<\z@ \let\final@kern1\fi
		\if\final@kern1 \kern-\dimen@\fi
		\else
		\overline{\rel@kern{-0.6}\kern\dimen@#1}%
		\fi
	}%
	\macc@depth\@ne
	\let\math@bgroup\@empty \let\math@egroup\macc@set@skewchar
	\mathsurround\z@ \frozen@everymath{\mathgroup\macc@group\relax}%
	\macc@set@skewchar\relax
	\let\mathaccentV\macc@nested@a
	\if#31
	\macc@nested@a\relax111{#1}%
	\else
	\def\gobble@till@marker##1\endmarker{}%
	\futurelet\first@char\gobble@till@marker#1\endmarker
	\ifcat\noexpand\first@char A\else
	\def\first@char{}%
	\fi
	\macc@nested@a\relax111{\first@char}%
	\fi
	\endgroup
}

\makeatletter
\def\ps@pprintTitle{%
	\let\@oddhead\@empty
	\let\@evenhead\@empty
	\def\@oddfoot{\footnotesize\itshape
		\ifx\@empty\@empty
		\else\@journal\fi\hfill\today}%
	\let\@evenfoot\@oddfoot	
}

\newcommand{\dd}{\mathrm{d}}
\newcommand{\upperRomannumeral}[1]{\uppercase\expandafter{\romannumeral#1}}
\newcommand{\lowerRomannumeral}[1]{\lowercase\expandafter{\romannumeral#1}}
\usepackage{indentfirst}
\usepackage{epstopdf}
\usepackage[caption=false]{subfig}
\theoremstyle{plain}

\newtheorem{hyp}{Hypothesis}
\newtheorem{theorem}{Theorem}
\newtheorem{lemma}[theorem]{Lemma}
\newtheorem*{lemma*}{Lemma}

\newtheorem{corollary}[theorem]{Corollary}
\newtheorem{proposition}[theorem]{Proposition}
\newtheorem{rem}{Remark}
\theoremstyle{definition}

\title{The rough Hawkes Heston stochastic volatility model\footnote{The research of Sergio Pulido benefited from the financial support of the chairs ``Deep finance \& Statistics'' and ``Machine Learning \& systematic methods in finance'' of \'Ecole Polytechnique. Sergio Pulido and Simone Scotti acknowledge support by the Europlace Institute of Finance (EIF) and the Labex Louis Bachelier, research project: ``The impact of information on financial markets''. }}

\author{Alessandro Bondi\thanks{Classe di Scienze, Scuola Normale Superiore di Pisa, alessandro.bondi@sns.it}\and Sergio Pulido\thanks{Universit\'e Paris-Saclay, CNRS, ENSIIE, Univ Evry, Laboratoire de Math\'ematiques et Mod\'elisation d'Evry (LaMME), sergio.pulidonino@ensiie.fr} \and Simone Scotti \thanks{Universit\`a di Pisa, simone.scotti@unipi.it}}

\begin{document}
	\maketitle
	
\begin{abstract}
We study an extension of the Heston stochastic volatility model that incorporates rough volatility and jump clustering phenomena. 
In our model, named the rough Hawkes Heston stochastic volatility model, the spot variance is a rough Hawkes-type process proportional to the intensity process of the jump component appearing in the dynamics of the spot variance itself and the log returns. 
The model belongs to the class of affine Volterra models. 
In particular, the Fourier-Laplace transform of the log returns and the square of the volatility index can be computed explicitly in terms of solutions of deterministic Riccati-Volterra equations, which can be efficiently approximated using a multi-factor approximation technique. 
We calibrate a parsimonious specification of our model characterized by a power kernel and an exponential law for the jumps. 
We show that our parsimonious setup is able to simultaneously capture, with a high precision, the behavior of the implied volatility smile for both S\&P 500 and VIX options. 
In particular, we observe that in our setting the usual shift in the implied volatility of VIX options is explained by a very low value of the power in the kernel. 
Our findings demonstrate the relevance, under an affine framework, of rough volatility and self-exciting jumps in order to capture the joint evolution of the S\&P 500 and VIX.

{\bf JEL code: C63, G12, G13}

{\bf Keywords: Stochastic volatility, Rough volatility, Hawkes processes, Jump clusters, Leverage effect,
affine Volterra processes, VIX, joint calibration of S\&P 500 and VIX smiles.}

\end{abstract} 

\section{Introduction}\label{intro}
%%%%%%% Volatility is stochastic. The VIX index. The joint calibration problem. Our approach. %%%%%%%
The Black-Scholes model, where volatility is constant, and more generally classical local volatility models, where volatility is a function of time and spot asset prices, fail to reproduce the dynamics of implied volatility smiles of options written on the underlying asset. 
To overcome this limitation, multiple stochastic, stochastic-local, and path-dependent volatility models have been developed and studied in recent years.
The complexity of volatility modeling, however, has increased with the significant growth over time of markets on volatility indices, such as the VIX.
The rise in popularity of these markets is explained in part by their relevance to protect portfolios \cite{Rho}. %, and to measure market stress \cite{Horst2} 
It has therefore become fundamental to develop stochastic models able to capture the joint dynamics of the underlying prices and their volatility index. 
The task is difficult because classical stochastic models fail to calibrate simultaneously the volatility smiles of options on the underlying and its volatility index. 
This modeling challenge, known as the joint S\&P 500/VIX calibration puzzle \cite{Guyon, Guyon2}, has inspired the introduction of more sophisticated models, e.g. \cite{Gatheral, Guyon2, Guyon3}, that incorporate new features to the joint dynamics of the underlying and the volatility in order to solve the problem.
In this paper we tackle the challenge by proposing a tractable affine model with rough volatility and volatility jumps that cluster and that have the opposite direction but occur at the same time as the jumps of the underlying prices. 
In this introduction we give a brief literature review to explain the choice of our framework.

%%%%%%% VIX spikes and clusters. Volatility jumps + little intro to joint calibration %%%%%%%
The dynamics of the VIX volatility index are highly complex. 
In particular, they exhibit large and systematically positive variations over very short periods, with a tendency to form clusters of spikes during difficult periods like the 2008 financial crisis and the beginning of the COVID-19 pandemic in 2020.
This is accompanied by very long periods without any large fluctuation and a less important mean reversion speed.
These observations are in line with an increasing number of studies that indicate the presence of jumps in the volatility \cite{Dotsis,Todorov}, on the underlying \cite{Bates}, and the fact these jumps are common to the volatility and underlying \cite{Sepp}.

The growing interest in volatility indices has driven the standardization of contingent claims written on the volatility indices themselves. 
These volatility index markets have very unique features. 
For VIX futures and Exchange-Traded products these features are studied in \cite{AP2018}. 
The complexity of volatility markets is also exemplified by the difficulty to jointly model the behavior of the volatility smiles of vanilla options written on the underlying and its volatility index, see for instance \cite{Alos, Papanicolaou,RITO}. 
This longstanding puzzle is known as the S\&P 500 (SPX)/VIX calibration puzzle. 
A growing body of literature explains the difficulty arguing that ``the state-of-the-art stochastic volatility models in the literature cannot capture the S\&P 500 and VIX option prices simultaneously'', see \cite{Song-Xiu}.  
As pointed out in \cite{Guyon, Guyon2}, ``all the attempts at solving the joint S\&P 500/VIX smile calibration problem only produced imperfect, approximate fits.''
The problem is that usual stochastic models either fail to reproduce one or both shapes of the implied volatility for S\&P 500 and VIX options or, when both the shapes are coherent, the implied volatility levels are incorrect.

%%%%%%%  Volatility memory non Markovian. Volaitlity is rough. %%%%%%%

Access to high frequency data has improved our understanding of the microstructure of financial markets and the effects on volatility. 
In particular, recent studies indicate that non-Markovian models with rough volatility trajectories might be appropriate to better capture long time dependencies due to meta orders and the large contribution of automatic orders. 
This is examined in \cite{Cont11} which provides a general analysis of order-driven markets, the work in \cite{ComteRenault} which elucidates the memory-features of volatility, and the studies in \cite{EEFR,Gath1} which give a micro-structural justification to the newly developed rough volatility models.

%%%%%%%  The Heston model. Its limitations.  rHeston%%%%%%%
From a modeling point of view, affine models provide a convenient framework because they are flexible and, thanks to semi-explicit formulas for the Fourier-Laplace transform, fast computations can be performed using Fourier-based techniques \cite{DFS2003, DPS2000, F01}. 
The most popular affine stochastic volatility model is the Heston model \cite{Heston}, where the spot variance is a square-root mean-reverting CIR (Cox-Ingersoll-Ross \cite{CIR85}) process. 
This model is able to reproduce some stylized features like the mean-reverting property of the volatility and the leverage effect. 
It is, however, unable to reproduce other phenomena such as extreme paths of volatility during crisis periods (even for large values of the volatility of volatility parameter) and the at the money (ATM) skews of underlying options' implied volatility simultaneously for short and long maturities. 
These limitations, and the micro-structural behavior of markets described in the previous paragraph, motivated the introduction of the rough Heston model \cite{rh1,rh2}.
The rough Heston model is tractable as it belongs to the class of affine Volterra models \cite{sergio}, and semi-explicit formulas for the Fourier-Laplace transform are still available. 
Unfortunately, this model cannot reproduce the features of options written on the volatility index and the underlying simultaneously.

 %%%%%%%  Description of the model. Rough + jumps + affine + implementation, fourier laplace multifactor and affine VIX in terms of vol%%%%%%%
In order to model the joint behavior of S\&P 500 and VIX markets, consistent with empirical evidence, we add two specific features to the usual Heston model. First, we incorporate rough volatility by adding a power kernel proportional to $t^{\alpha-1}$, with $\alpha\in (1/2,1]$, to the dynamics of the spot variance. Second, we postulate common jumps for the volatility and the underlying with a negative leverage. 
The presence of jumps in both underlying and variance helps to reproduce a skewed implied volatility for vanilla options as in the Barndorff-Nielsen and Shephard model \cite{BarSheInBook, BarShe}. 
Inspired by the Hawkes framework, taking into account jump-clustering and endogeneity of financial markets, we model the spot variance to be proportional to the intensity process of the jump component appearing in the dynamics of the spot variance itself and the log returns. 
For these reasons, we name our model the rough Hawkes Heston model.

To keep mathematical and numerical tractability, we choose an affine specification of the model.
 As such, our model belongs to the class of affine Volterra processes \cite{sergio}, which has been recently extended to jump processes in \cite{primo, CT, CT2}. In particular, the Fourier-Laplace transform of the log returns and the square of the volatility index can be computed explicitly in terms of solutions of deterministic Riccati-Volterra equations, see Theorems \ref{t6} and \ref{t6V}. 
 We approximate the solutions of the Riccati-Volterra equations via a multi-factor scheme as in \cite{ee}, see Theorem \ref{stime}. 
 We leave for future study the implementation and analysis in our framework of other methods such as the Adams method \cite{Ada1, Ada2}, asymptotic expansions based on Malliavin calculus in the spirit of \cite{Alos-GR}, and hybrid approximation techniques for Volterra equations similar to those in \cite{CGP}. 
 
 The affine property is an advantage of our modeling approach compared to other models proposed to solve the SPX/VIX calibration problem, such as the quadratic rough Heston model \cite{Gatheral}, where pricing is done via Monte Carlo or machine learning techniques \cite{rosenbaum_zhang}. 
 In addition, our affine framework is convenient because Variance Swap prices and the square VIX index have explicit affine relations to the forward curve, see Corollary \ref{corollary-affine} and Remark \ref{Var_Swap}. 
This is a generalization, to the affine Volterra setting, of the affine relation already pointed out in \cite{Kallsen} within the classical affine exponential framework and empirically confirmed in \cite{MST20}.

 %%%%%%%  The joint calibration problem via jumps and other modeling approaches  %%%%%%%

Previous literature on jump-diffusion models focusing on the evolution of S\&P 500 and the VIX proposes either high-dimensional models \cite{ContKok13, PPR, Sepp}, or models based on hidden Markov chains \cite{GIP, PS}. 
These models require a large number of parameters and suffer from the lack of interpretability of the random factors. 
Our approach to model the joint SPX/VIX dynamics is different. As in \cite{BBSS}, we keep the number of parameters low by assuming that the jump intensity is proportional to the variance process itself, and jumps are common to the volatility and underlying with opposite signs. 
The main new ingredient of our model, compared to \cite{BBSS}, is the addition of a Brownian component and a power kernel to the variance process. 
This generates by construction a jump clustering effect and takes into account related findings in the rough volatility literature \cite{AlosLeonVives, Alos-S, Bayer, bennedsen_pakkanen,EEFR, rh2, Fuka,  Gath1, Gatheral,Livieri}. 

 %%%%%%%  The parameter alpha in our model close to 1/2. ATM skew  %%%%%%%

The rough Hawkes Heston model is able to reconcile the shapes and level of the S\&P 500 and VIX volatility smiles. 
An important role is played by the parameter $\alpha$ characterizing the kernel. 
As is the case for other rough volatility models, this parameter controls the explosion rate of the term structure of ATM skews for SPX option smiles as maturity goes to zero. 
We show that when $\alpha$ is near to $1/2$, the rate of explosion is in the range $[0.5, 0.6]$. 
This is consistent with similar findings in the rough volatility literature \cite{Alos-S, Bayer, bennedsen_pakkanen,EEFR, Fuka, Gath1, Gatheral}. 
In addition, in our framework, the parameter $\alpha$ plays a crucial role because it controls the level of the implied volatility of VIX options for short maturities. We observe, that as $\alpha$ approaches $1/2$ the levels of S\&P 500 and VIX smiles are coherent.

 %%%%%%% Summary  %%%%%%%
To summarize, the model that we propose in this paper shares many features with other existing models. 
These features are mainly: 
rough volatility \cite{Bayer, EEFR,  rh2, Fuka, Gath1, Gatheral}, jumps \cite{BarSheInBook,BarShe, Bates, ContKok13, PPR, Sepp}, the Hawkes/branching character of volatility \cite{BBSS, BrignoneSgarra, Horst2}, and the affine structure \cite{sergio, primo, DFS2003, DPS2000, F01, Kallsen, JMSZ}. 
Consequently we take advantage of the low regularity and memory features of rough volatility models, the large fluctuation of jumps, the clusters of Hawkes processes and the explicit Fourier-Laplace transform of the affine setup.
The specification that we adopt for the joint SPX/VIX calibration is parsimonious with only five evolution-related parameters. 
Moreover, all the parameters have a financial interpretation. 
The parameter $\alpha$ in the kernel controls the decay of the volatility memory, SPX ATM skews and the level of VIX smiles. 
We have in addition the classical parameters controlling the volatility mean reversion speed and the volatility of volatility, and two parameters related to the leverage effect that specify the correlation between Brownian motions and between the jumps in the asset and its volatility. 
Despite its robustness, the rough Hawkes Heston stochastic volatility model captures remarkably well the implied volatility surfaces of S\&P 500 and VIX at the same time.

 %%%%%%% Oraganization  %%%%%%%
The paper is organized as follows. 
Section \ref{s_model} lays out the essential hypotheses of our study and introduces the stochastic model under a general setup, i.e. with a general kernel and law for the jumps.
Section \ref{FL_logS} explains the derivation of the Fourier-Laplace transform of the log returns and the application to undelying's options pricing.
Section \ref{sect-VIX2} focuses on the VIX index characterizing the Fourier-Laplace transform of the VIX$^2$, and describes the Fourier-based formulas to price options on the VIX. 
Section \ref{sec_num} studies the multi-factor numerical scheme used in order to approximate the solutions to the Riccati-Volterra equations arising in Sections \ref{FL_logS} and \ref{sect-VIX2}.
Section \ref{sec_cal} details the calibration of our model to S\&P 500 and VIX options data.
Section \ref{sec_sensi} presents a complete and detailed sensitivity analysis of implied volatility curves with respect to the model parameters.
Section \ref{sec_conclusion} summarizes the conclusions of our study.  
Appendix \ref{ap_A} contains the proof of the necessary existence, uniqueness and comparison results for the Riccati-Volterra equations appearing in Section \ref{FL_logS}.
Appendix \ref{Lew} presents the proof of the Fourier-inversion formula used to price options on the underlying.
To finish, in Appendix \ref{ap_C} we prove the main result related to the convergence of the multi-factor approximation scheme for the Riccati-Volterra equations.

	\section{The model}\label{s_model}
	
	We study a stochastic volatility model where the spot variance  $\sigma^2=(\sigma^2_t)_{t\ge0}$ is a predictable process, with trajectories in $L^2_\text{loc}(\mathbb{R}_+)$, defined on a stochastic basis $(\Omega,\mathcal{F},\mathbb{Q},\mathbb{F}=(\mathcal{F}_t)_{t\ge0})$. We assume that the filtration $\mathbb{F}$ satisfies the usual conditions and that $\mathcal{F}_0$ is the trivial $\sigma-$algebra. 
	
	We consider, throughout our study, a kernel $K$ that satisfies the next requirement, see \cite{edu, ACLP, sergio, primo}.

	\begin{hyp}\label{c1}
		The kernel $K\in L^2_\emph{loc}(\mathbb{R}_+)$ is nonnegative, nonincreasing, not identically zero and continuously differentiable on $(0,\infty)$. Furthermore, its resolvent of the first kind $L$ exists and it is nonnegative and nonincreasing, i.e. $ s\mapsto L[s,s+t]$ is nonincreasing for every $t\ge0$. 
	\end{hyp}
	We recall that, given a kernel $K\in L^1_{\text{loc}}(\mathbb{R}_+; \mathbb{R}^{d\times d})$, an $\mathbb{R}^{d\times d}-$valued measure $L$ is called its (measure) resolvent of the first kind if $L\ast K=K\ast L=I$, where $I\in\mathbb{R}^{d\times d}$ is the identity matrix. The resolvent of the first kind does not always exist, but if it does then it is unique, see \cite[Theorem $5.2$, Chapter $5$]{g}. 
	
	We assume that the spot variance $\sigma^2$ is a $\mathbb{Q}\otimes \dd t-$a.e. nonnegative predictable process which satisfies the following stochastic affine Volterra equation of convolution type with jumps:
	\begin{equation}\label{variance}
		\sigma^2=g_0+K\ast \dd Z,\quad \mathbb{Q}\otimes \dd t-\text{a.e.}
	\end{equation}
	Here $Z=(Z_t)_{t\ge0}$ is a semimartingale starting at $0$ with associated jumps-measure $\mu(\dd t,\dd z)$ and compensated measure $\tilde{\mu}(\dd t,\dd z)=\mu(\dd t,\dd z)-\nu(\dd z)\sigma^2_t\,\dd t$, with $\nu$ a nonnegative measure on $\mathbb{R}_+$  such that  $\nu(\{0\})=0$ and $\int_{\mathbb{R}^{+}}|z|^2\nu(\dd z)<\infty$. Since the intensity of the jumps of $\sigma^2$ is proportional to $\sigma^2$ itself, the spot variance is a Hawkes-type process, which is coherent with other models that incorporate endogeneity of financial markets such as \cite{BBSS, CMS, rh2, Gonzato,JMS}. 
	More specifically, $Z$ is given by
	\[
	\dd Z_t=b\,\sigma^2_t\dd t+\sqrt{c}\,\sigma_t \,\dd W_{2,t} + \int_{\mathbb{R}_{+}}z\,\tilde{\mu}\left(\dd t,\dd z\right),\quad Z_0=0,
	\]
	where $b\in\mathbb{R},\,c>0$ and $W_2=(W_{2,t})_{t\ge 0}$ is an $\mathbb{F}-$Brownian motion. 
	%Here $\sigma=\left(\sigma_t\right)_{t\ge 0}$ is the predictable process defined by $\sigma_t=\sqrt{\sigma^2_t1_{\left\{\sigma^2\ge 0\right\}}}$. 
	In the sequel, we denote by $\widetilde{Z}=(\widetilde{Z}_t)_{t\ge0}$ the process $\widetilde{Z}_t=\sqrt{c}\,\sigma_t\,\dd W_{2,t}+\int_{\mathbb{R}_{+}}z\,\tilde{\mu}(\dd t,\dd z),\,t\ge0$. Notice that $\widetilde{Z}$ is a square-integrable martingale by \cite[Lemma $1$]{primo}. 
	The function $g_0$ is the initial input spot variance curve. By analogy with the rough Heston model introduced and studied in \cite{rh1,rh2}, we consider it of the form
	\begin{equation}\label{in_c}
		g_0\left(t\right)=\sigma_0^2+\beta\int_{0}^{t}K\left(s\right)\dd s,\quad t\ge 0,
	\end{equation}
	where $\sigma_0^2,\,\beta\ge0$. According to \cite[Appendix A]{primo}
	\begin{equation}\label{canonical}
		\sigma^2=g_0-R_{-bK}\ast g_0+E_{b,K}\ast \dd \widetilde{Z},\quad \mathbb{Q}\otimes \dd t-\text{a.e.},
	\end{equation}
	where $R_{-bK}$ is the resolvent of the second kind of $-bK$ and $E_{b,K}$ is the canonical resolvent of $K$ with parameter $b$. We recall that the resolvent of the second kind $R_K$ for a kernel $K\in L^1_{\text{loc}}(\mathbb{R}_+)$  is the unique solution $R_K\in L^1_{\text{loc}}(\mathbb{R}_+)$ of the two equations $K\ast R_K=R_K\ast K=K-R_K$. The canonical resolvent $E_{\lambda,K}$ of $K$ with parameter $\lambda$	is defined by $E_{\lambda,K}=-\lambda^{-1}R_{-\lambda K}$ for $\lambda\neq0$, whereas $E_{0,K}=K$,
	see \cite[Theorem 3.1, Chapter 2]{g} and the subsequent definition.
	
	\begin{rem}
	If we assume that $K$  and the shifted kernels $K(\cdot+1/n)$, $n\in\mathbb{N}$, satisfy Hypothesis \ref{c1}, the (weak) existence of the spot variance process $\sigma^2$, satisfying \eqref{variance}, is ensured by \cite[Theorem $2.13$]{edu} and \cite[Lemma $9$]{primo}. Assuming weak existence, weak uniqueness is established in \cite[Corollary $12$]{primo} under Hypothesis \ref{c1}. We refer to \cite{ACLP} and \cite{primo} for more information about stochastic Volterra equations and stochastic convolution for processes with jumps.
	\end{rem}
	
	A useful tool for the development of the theory is the adjusted forward process, which we now define. For every $t\ge0$, it is denoted by $(g_t(s))_{s>t}$ and it is a jointly measurable process on $\Omega\times (t,\infty)$ such that
	\begin{equation}\label{adjusted}
		g_t\left(s\right)= g_0\left(s\right)+\int_{0}^{t}K\left(s-r\right)\dd Z_r,\quad \mathbb{Q}-\text{a.s., }s>t.
	\end{equation}
	Thanks to \cite[Theorem 46]{meyer} and the fact that $\mathbb{F}$ satisfies the usual conditions, we can consider $g_t(\cdot)$ to be $\mathcal{F}_t\otimes \mathcal{B}(t,\infty)-$measurable.  \\
	%Note that for $t=0$ we have an abuse of notation, as $g_0$ represents both the initial input curve in \eqref{in_c} and the process just defined in \eqref{adjusted}. This, however, is not an issue as these two concepts coincide $\mathbb{Q}\otimes \dd t-$a.e. in $\Omega\times \left(0,\infty\right)$. In the following, we continue to consider $g_0$ as the initial input curve.\\
	Analogous arguments provide a version of the conditional expectation process $\mathbb{E}[\sigma^2|\mathcal{F}_t]=(\mathbb{E}[\sigma^2_s|\mathcal{F}_t])_{s>t}$ which is $\mathcal{F}_t\otimes \mathcal{B}(t,\infty)-$measurable. In particular, from \eqref{canonical},
	\begin{equation}\label{forward}
		\mathbb{E}\left[\sigma^2_s\Big| \mathcal{F}_t\right]=g_0-R_{-bK}\ast g_0+\int_{0}^{t}E_{b,K}\left(s-r\right)\dd \widetilde{Z}_r,\quad \mathbb{Q}-\text{a.s., }s>t.
	\end{equation}
	
	\begin{comment}
		The connection between the two processes just introduced is given by the next equation
		\begin{multline}\label{lias}
			g_t\left(s\right)=\mathbb{E}\left[\sigma^2_s-b\int_{t}^{s}K\left(s-r\right)\sigma^2_{r}\,\dd r\,\Big|\,\mathcal{F}_t\right]\\
			=\mathbb{E}\left[\sigma^2_s\Big|\, \mathcal{F}_t\right]-b\int_{t}^{s}K\left(s-r\right)\mathbb{E}\left[\sigma^2_r\Big|\,\mathcal{F}_t\right]\dd r
			,\quad \mathbb{Q}-\text{a.s., for a.e. }s>t,
		\end{multline}
		where the equality of the first and last term can be understood in the $\mathbb{Q}\otimes \dd t-$a.e. sense.
	\end{comment}
	
	We now prescribe the dynamics of the log returns process $X=(X_t)_{t\ge0}$ as follows:
	\begin{multline}\label{spot}
		\dd X_t=-\left(\frac{1}{2}+\int_{\mathbb{R}_+}\left(e^{-\Lambda z}-1+\Lambda z\right)\nu\left(\dd z\right)\right)\sigma^2_t\,\dd t+\sigma_t\left(\sqrt{1-\rho^2}\,\dd W_{1,t}+\rho\,\dd W_{2,t}\right)\\-\Lambda\int_{\mathbb{R}_+}z\,\tilde{\mu}\left(\dd t, \dd z\right),\quad X_0=0,
	\end{multline}
	where $\rho\in[-1,1]$ is a correlation parameter, 
	$W_1=(W_{1,t})_{t\ge0}$ is an $\mathbb{F}-$Brownian motion independent from $W_2$ and
	$\Lambda\ge0$ is a leverage parameter forcing common jumps for volatility and underlying with opposite signs. 
	This is coherent with empirical findings in \cite{Todorov}, stylized features studied 
	in \cite{Cont01}, and the financial/econometric 
	literature with jumps, e.g. \cite{BarSheInBook, BarShe, Bates, BBSS, CS15, RST22, Sepp}.
	We have assumed, for the  sake of readability and without lost of generality, that interest rates are zero.
	The price process of the underlying asset will be $S=(S_t)_{t\ge 0}=(S_0e^{X_t})_{t\ge0}$, where $S_0>0$ represents the initial price. An application of It\^o's formula shows that $S$ is a local martingale. Indeed,
	\begin{align*}
		\frac{\dd S_t}{S_{t-}}&=-\left(\frac{1}{2}+\int_{\mathbb{R}_+}\left(e^{-\Lambda z}-1+\Lambda z\right)\nu\left(\dd z\right)\right)\sigma_t^2\dd t+\sigma_t\left(\sqrt{1-\rho^2}\,\dd W_{1,t}+\rho\,\dd W_{2,t}\right)\\ &\quad
		-\Lambda \int_{\mathbb{R}_+}z\,\tilde{\mu}\left(\dd t, \dd z\right) +\frac{1}{2}\sigma_t^2\,\dd t+\int_{\mathbb{R}_+}\left(e^{-\Lambda z}-1+\Lambda z\right)\,\mu\left(\dd t,\dd z\right)\\
		&=\sigma_t \left(\sqrt{1-\rho^2}\,\dd W_{1,t}+\rho\,\dd W_{2,t}\right)+ \int_{\mathbb{R}_{+}}\left(e^{-\Lambda z}-1\right)\tilde{\mu}\left(\dd t,\dd z\right)
		\eqqcolon\dd N_t,
	\end{align*}	
	where $N=(N_t)_{t\ge0}$ is a local martingale with $N_0=0$. In particular, since $S$ starts at $S_0$, it follows that $S=S_0\mathcal{E}(N)$, where $\mathcal{E}$ denotes the Dol\'eans-Dade exponential. In the next section, see Corollary \ref{cor_mar}, we will improve on this result by showing that, for every $T>0$, the restriction of $S$ to $[0,T]$ is a true martingale.
	% In any case, note that we are working under a martingale measure $\mathbb{Q}$ with risk--free interest rate $r=0$. 
	
	\section{The Fourier-Laplace transform of the log returns}\label{FL_logS}
	In this section we study, for a fixed $T\ge0$, the conditional Fourier-Laplace transform of $X_T$, $\mathbb{E}[e^{wX_T}|\,\mathcal{F}_t]$, $t\in[0,T]$. Here $w\in\mathbb{C}$ is subject to suitable conditions that will be specified in the sequel. In particular, we want to find a formula that allow us to compute the prices of options written on the underlying asset using Fourier-inversion techniques \cite{DFS2003, DPS2000, F01,G2016}. We will adopt the following notation: for $z\in\mathbb{C}$ we denote by $\mathfrak{R} z$ and $\mathfrak{Im} z$ the real and imaginary parts of $z$, respectively. We let $\mathbb{C}_+$ [resp., $\mathbb{C}_-$] be the set of complex real numbers with nonnegative [resp., nonpositive] real part.
	
	Let us define the mapping $F\colon \mathbb{C}_+\times\mathbb{C}_-\to\mathbb{C}$ by
	\begin{equation}\label{F_log}
		F\left(u,v\right)=\frac{1}{2}\left(u^2-u\right)+\left(b+\rho\sqrt{c}\,u\right)v+\frac{c}{2}v^2
		+\int_{\mathbb{R}_+}\left[e^{\left(v-\Lambda u\right)z}-u\left(e^{-\Lambda z}-1\right)-1-vz\right]\nu\left(\dd z\right),
	\end{equation}
	for every $(u,v)\in\mathbb{C}_+\times \mathbb{C}_-$. For the development of the theory we need the following result about deterministic Riccati-Volterra equations, whose proof is postponed to Appendix \ref{ap_A}.
	\begin{theorem}\label{gl_log}
		Suppose that $K$ satisfies Hypothesis \ref{c1} and $w\in\mathbb{C}$ is such that  $\mathfrak{R}w\in[0,1]$.
		\begin{enumerate}[label=\emph{(\roman*)}]
			\item\label{1t1} There exists a unique continuous solution $\psi_w\colon\mathbb{R}_+\to\mathbb{C}_-$ of the Riccati-Volterra equation
			\begin{equation}\label{RV_log}
				\psi_w\left(t\right)=\int_{0}^{t}K\left(t-s\right) F\left(w, \psi_w\left(s\right)\right)\dd s=\left(K\ast F\left(w,\psi_w\left(\cdot\right)\right)\right)\left(t\right),\quad t\ge0.
			\end{equation}
			In particular, $\psi_{\mathfrak{R}w}$ is $\mathbb{R}_--$valued.
			\item\label{2t1}  The following inequalities hold:
			\begin{equation}\label{comp_log}
				\mathfrak{R}\psi_w\left(t\right)\le {\psi_{\mathfrak{R}w}}\left(t\right)\le 0,\quad t\ge 0.
			\end{equation}
		\end{enumerate} 
	\end{theorem}
	\noindent We also need the next preparatory lemma, which can be proven similarly to\,\cite[Lemma $6.1$]{edu}.
	\begin{lemma}\label{prep}
		Let $f_1,\,f_2,\,f_3\colon[0,T]\to\mathbb{R}$ be bounded measurable functions such that $f_3\le0$ in $[0,T]$. Then, the Dol\'eans-Dade exponential
		\[
		\mathcal{E}\left(\int_{0}^{t}f_1\left(s\right)\sigma_s\,\dd W_{1,s}+\int_{0}^{t}f_2\left(s\right)\sigma_s\,\dd W_{2,s}+ \int_{0}^{t}\int_{\mathbb{R}_{+}}\left(e^{f_3\left(s\right) z}-1\right)\tilde{\mu}\left(\dd s,\dd z\right)\right),\quad t\in\left[0,T\right]
		\]
		is a martingale.
	\end{lemma} 
	We are now ready to state the main result of this section. We introduce for every $\varepsilon\in\mathbb{R}$ the shift operator $\Delta_{\varepsilon}$, which, given $I\subset \mathbb{R}$ and a function $f\colon I\to\mathbb{C}$, assigns the function $\Delta_{\varepsilon}f\colon I-{\varepsilon}\to \mathbb{C}$ defined by $\Delta_{\varepsilon}f(t)=f(t+{\varepsilon}),\,t\in I-{\varepsilon}$.
	\begin{theorem}\label{t6}
		Suppose that $K$ satisfies Hypothesis \ref{c1}  and that the resolvent of the first kind $L$ is the sum of a locally integrable function and a point mass at $0$. Moreover, suppose that the total variation bound 
		\begin{equation*}
			\sup_{{\varepsilon}\in\left(0,\widebar{T}\right]}\norm{\Delta_{\varepsilon}K \ast L}_{\text{TV}\left(\left[0,\widebar{T}\right]\right)}<\infty
		\end{equation*} 
		holds for all $\widebar{T}>0$. Then, for every $w\in\mathbb{C}$ such that $\mathfrak{R}w\in[0,1],$ 
		\begin{equation}\label{final}
			\mathbb{E}\left[\exp\left\{wX_T\right\}\Big|\mathcal{F}_t\right]
			=
			\exp\left\{\widetilde{V}_t\left(w,T\right)\right\}
			,\quad \mathbb{Q}-\text{a.s., }t\in\left[0,T\right],
		\end{equation}
		where $\widetilde{V}_t(w,T)=wX_t+\int_{t}^{T}F(w,\psi_w(T-s))g_t(s)\dd s,\,t\in[0,T]$.
	\end{theorem}
	\begin{proof}
		
		Let $w\in\mathbb{C}$  be such that $\mathfrak{R}w\in[0,1]$. Define the c\`adl\`ag, adapted, $\mathbb{C}-$valued semimartingale  $(V_t(w,T))_{t\in[0,T]}$ by
		\begin{align}
			\begin{split}
				V_t\left(w,T\right)&=V_0\left(w,T\right)+wX_t+\int_{0}^{t}\psi_w\left(T-s\right) \dd \widetilde{Z}_s 
				\\&\quad 
				-\int_{0}^t\Bigg(\frac{1}{2}\left(w^2-w\right)+\rho\sqrt{c}\,w\,\psi_w\left(T-s\right) 
				+\frac{c}{2}\psi_w\left(T-s\right)^2 \\
				&\quad 
				+\int_{\mathbb{R}_+}\!\left(e^{\left(-\Lambda w+\psi_w\left(T-s\right)\right)z}-w\left(e^{-\Lambda z}-1\right)-1-\psi_w\left(T-s\right)z\right)\nu\left(\dd z\right)\Bigg)\sigma^2_s\,\dd s,\label{y1}
			\end{split}
			\\
			\begin{split}
				V_0\left(w,T\right)&=\int_{0}^{T}F\left(w,\psi_w\left(T-s\right)\right) g_0\left(s\right)\dd s.
				\label{y2_2}
			\end{split}
		\end{align}
		The same arguments as in the proof of  \cite[Theorem $5$]{primo}, which essentially rely on the stochastic Fubini's theorem (see, e.g., \cite[Theorem $65$, Chapter \upperRomannumeral{4}]{protter}),  allow us  to prove that 
		\begin{equation}\label{thj4.3}
			V_t\left(w,T\right)=\widetilde{V}_t\left(w,T\right),\quad \mathbb{Q}-\text{a.s.},\,  t\in\left[0,T\right].
		\end{equation}
		We now define $H(w,T)=(H_t(w,T))_{t\in[0,T]}=(\exp\{V_t(w,T)\})_{t\in[0,T]}$. By  {It\^o's formula} and the dynamics in \eqref{spot} and  \eqref{y1} we have, omitting $(w,T)$ for sake of readability, 
		\begin{align*}
			\frac{\dd H_{t}}{H_{t-}}&=\Bigg[w\,\dd X_t-\Bigg(\frac{c}{2}\,\psi_w\left(T-t\right)^2
			+	\int_{\mathbb{R}_+}\left(e^{\left(-\Lambda w+\psi_w\left(T-t\right)\right)z}-1 -w\left(e^{-\Lambda z}-1\right)-\psi_w\left(T-t\right)z\right)\nu\left(\dd z\right)
			\\&\quad+\frac{1}{2}\left(w^2-w\right)+\rho\sqrt{c}\,w\,\psi_w\left(T-t\right)
			\Bigg)\sigma^2_t\,\dd t 
			+\psi_w\left(T-t\right) \dd \widetilde{Z}_t\Bigg]
			+\frac{1}{2}\left({c}\,\psi_w\left(T-t\right)^2+w^2\right)\sigma^2_t\,\dd t
			\\&\quad 
			+\rho\sqrt{c}\,w\,\psi_w\left(T-t\right)\sigma^2_t\,\dd t
			+\int_{\mathbb{R}_+}\left(e^{\left(-\Lambda w+\psi_w\left(T-t\right)\right) z}-1-\left(-\Lambda w+\psi_w\left(T-t\right)\right)z\right)\mu\left(\dd t,\dd z\right)\\
			&=\left[\sigma_t\left(w\sqrt{1-\rho^2}\,\dd W_{1,t}+\left(w\rho+\sqrt{c}\,\psi_w\left(T-t\right)\right) \dd W_{2,t}\right)
			+\int_{\mathbb{R}_+}\left(e^{\left(-\Lambda w+\psi_w\left(T-t\right)\right) z}-1\right)\tilde{\mu}\left(\dd t,\dd z\right)\right], 
		\end{align*}
		with $H_0=\exp(V_0)$. We define $N(w,T)=(N_t(w,T))_{t\in[0,T]}$ by $N_0(w,T)=0$ and 
		\begin{equation*}
			\begin{array}{rcl}
				\displaystyle \dd N_t\left(w,T\right)&=& \displaystyle \sigma_t\left(w\sqrt{1-\rho^2}\,\dd W_{1,t}+\left(w\rho+\sqrt{c}\,\psi_w\left(T-t\right)\right) \dd W_{2,t}\right) \\
				&& \displaystyle+\int_{\mathbb{R}_+}\left(e^{\left(-\Lambda w+\psi_w\left(T-t\right)\right)z}-1\right)\tilde{\mu}\left(\dd t,\dd z\right).
			\end{array}
		\end{equation*}
		Then $N(w,T)$ is a local martingale and the previous computations show that, omitting again $(w,T)$, $H=\exp\{V_0\}\mathcal{E}(N)$ up to evanescence, where $\mathcal{E}$ denotes the Dol\'eans-Dade exponential. Therefore $H(w,T)$ is a local martingale. If it is indeed a true martingale, then \eqref{final} directly follows from \eqref{thj4.3} noting also that $\widetilde{V}_T(w,T)=wX_T$.
		
		In order to argue the martingale property of $H(w,T)$, first we observe that by Lemma \ref{prep} the real-valued process $H(\mathfrak{R}w,T)=(H_t(\mathfrak{R}w,T))_{t\in[0,T]}=(\exp\{V_t(\mathfrak{R}w,T)\})_{t\in[0,T]}$ is a true martingale. Secondly, we invoke \cite[Corollary $8$]{primo} to obtain the following alternative expression for $V(w,T)$ (an analogous one holds for $V(\mathfrak{R}w, T)$)
		\begin{multline}\label{past}
			V_t\left(w,T\right)=
			wX_t+\int_{0}^{T-t}F\left(w,\psi_w\left(s\right)\right)g_0\left(T-s\right)\dd s
			+\psi_w\left(T-t\right)L\left(\left\{0\right\}\right)\left(\sigma^2-g_0\right)\left(t\right)
			\\+\left(\dd {\Pi}_{T-t}\ast\left(\sigma^2-g_0\right)\right)\left(t\right)
			,\quad \text{for a.e. }t\in\left(0,T\right),\,\mathbb{Q}-\text{a.s.},
		\end{multline}
		where for every ${\varepsilon}>0$, $\Pi_{\varepsilon}(t)=\int_{0}^{\varepsilon}F(w,\psi_w(s)) (\Delta_{{\varepsilon}-s}K\ast L)(t)\dd s,\,t\ge0,$ is a locally absolutely continuous function. The application of this result is legitimate because the procedure carried out in \cite{primo} to infer \eqref{past} only depends on \eqref{variance}, \eqref{RV_log} and the boundedness on compact intervals of $\mathbb{R}_+$ of $F(w,\psi_w(\cdot))$, and does not rely on the expression of $F$. A similar argument together with \eqref{comp_log} allows us to parallel the comparison method in the proof of  \cite[Theorem $11$]{primo} to conclude that there is a constant $C>0$ such that
		\[
		\left|H_t\left(w,T\right)\right|=\left|\exp\left\{V_t\left(w,T\right)\right\}\right|=\exp\left\{\mathfrak{R}V_t\left(w,T\right)\right\}\le C\exp\left\{V_t\left(\mathfrak{R}w,T\right)\right\}=CH_t\left(\mathfrak{R}w,T\right),
		\]
		for $t\in[0,T]$, $\mathbb{Q}-\text{a.s.}$ At this point it is sufficient to invoke \cite[Lemma $1.4$]{J} to claim that $H(w,T)$ is a true martingale, hence the proof is complete.
	\end{proof}
	From the previous theorem we deduce the martingale property of our price process $S$ with a direct approach (it can also be obtained by Lemma \ref{prep}).
	\begin{corollary}\label{cor_mar}
		Under the hypotheses of Theorem \ref{t6}, the price process $S=(S_t)_{t\in[0,T]}$ is a martingale.
	\end{corollary}
	\begin{proof}
		The computations at the end of Section \ref{s_model} show that the stock price $S$ is a nonnegative local martingale, hence it is a supermartingale. In order for it to be a martingale, it is sufficient to show that $\mathbb{E}[S_T]=S_0$. By \eqref{final} in Theorem \ref{t6} with $w=1$ we have
		\[
		\mathbb{E}\left[S_T\right]=S_0\exp\left\{\int_{0}^{T}F\left(1,\psi_1\left(t-s\right)\right)g_0\left(s\right)\dd s\right\}.
		\]
		From \eqref{F_log}-\eqref{RV_log}, we observe that $\psi_1\equiv0$ in $\mathbb{R}_+$. This implies that $F(1,\psi_1(\cdot))=0$ in $\mathbb{R}_+$, which concludes the proof.
	\end{proof}
	Equation \eqref{final} in Theorem \ref{t6} gives a semi-explicit expression to compute the Fourier-Laplace transform $\Psi^{X_T}$ of $X_T$ in a suitable region of $\mathbb{C}$, namely
	\begin{equation}\label{LP_X}
		\Psi^{X_T}\left(w\right)=\exp\left\{\int_{0}^TF\left(w,\psi_w\left(T-s\right)\right)g_0\left(s\right)\,\dd s\right\},\quad w\in\mathbb{C} \text{ such that }\mathfrak{R}w\in\left[0,1\right].
	\end{equation}
	As shown in the following proposition, whose proof is in Appendix \ref{Lew}, we can use $\Psi^{X_T}$ to price options with maturity $T$ on the underlying asset $S$ via Fourier-inversion techniques.	
	\begin{proposition}\label{price_Lewis}
		Fix a log strike $k>0$. Then, under the hypotheses of Theorem \ref{t6}, the price $C_S(k,T)$ of a call option on the underlying asset $S$ with log strike $k$ and maturity $T$ is
		\begin{equation}\label{pr_SPX}
			C_S\left(k,T\right)=S_0-\frac{1}{\pi}\sqrt{S_0 e^k}\int_{\mathbb{R}_+}\mathfrak{R}\left[e^{i\lambda\left(\log(S_0)-k\right)}\Psi^{X_T}\left(\frac{1}{2}+i\lambda\right)\right]\frac1{\frac{1}{4}+\lambda^2}\,\dd \lambda,
		\end{equation}
		and the price $P_S(k,T)$ of a put option with the same log strike, maturity and underlying is 
		\begin{equation}\label{put_SPX}
			P_S\left(k,T\right)=e^k-\frac{1}{\pi}\sqrt{S_0e^k}\int_{\mathbb{R}_+}\mathfrak{R}\left[e^{i\lambda\left(\log(S_0)-k\right)}\Psi^{X_T}\left(\frac{1}{2}+i\lambda\right)\right]\frac1{\frac{1}{4}+\lambda^2}\,\dd \lambda.
		\end{equation}
	\end{proposition}
	\begin{rem}
		The expression in \eqref{pr_SPX} coincides with \cite[Formula $(3.11)$]{lewis}, but we have to independently prove it (see Appendix \ref{Lew}). Indeed, in \cite{lewis} the author obtains \eqref{pr_SPX} starting from the inversion of the generalized Fourier transform of the payoff function $w(x)=(e^x-e^k)^+,\,x\in\mathbb{R},$ of a call option with log strike $k$ (here $x$ represents the log price). Namely, for $x\in\mathbb{R},$
		\[
		w\left(x\right)=-\frac{1}{2\pi}\int_{iz_i-\infty}^{iz_i+\infty}\frac{e^{k(iz+1)}}{z^2-iz}e^{-izx}\dd z,\quad z_i>1.
		\]
		If we were to follow the same approach here, then we would find a problem: we only have proved that $\Psi^{X_T}$ is defined for complex numbers with real part in  $[0,1]$. Therefore, in the previous expression, we would need $z_i\in[0,1]$, which is a contradiction. This setback cannot be immediately fixed by considering put options and then applying the put-call parity formula, because again the intersection between the complex strip ($z_i<0$), where the Fourier transform for the payoff function is defined, and the strip where $\Psi^{X_T}(-i\,\cdot)$ is available is empty. We refer to \cite[Section $4$]{Sc} for a survey of pricing based on Fourier-inversion techniques. 
	\end{rem}
	
	\section{The Fourier-Laplace transform of VIX$^2$}\label{sect-VIX2}
	In this section the underlying asset $S$ represents the SPX index. Then, according to the CBOE VIX white paper and  \cite{more}, the theoretical value of VIX=$(\text{VIX}_T)_{T\ge0}$ is 
	\begin{equation}\label{VIX}
		\text{VIX}_T=\sqrt{\left(-\frac{2}{{\delta}}\mathbb{E}\left[ X_{T+{\delta}}-{X_{T}}\big| \mathcal{F}_T\right]\right)^+}\times 100,\quad T\ge0.
	\end{equation}
	Here ${\delta}=\frac{1}{12}$ and represents $30$ days, the time to expiration of the log contracts involved in the computation of the index. Note that in \eqref{VIX}, the positive part has been inserted to guarantee the good definition of the random variable $\text{VIX}_T$ in the whole space $\Omega$, however the radicand is nonnegative $\mathbb{Q}-$a.s., as we are about to show.
	
	We first derive, in the following theorem, an expression for $\mathbb{E}[ X_{T+{\delta}}-{X_{T}}| \mathcal{F}_T]$, $T\ge 0$, in terms of the adjusted forward process at time $T$, $g_T(\cdot)$.
	
	\begin{theorem}\label{theo-affine-VIX2}
		The log contract satisfies an infinite dimensional affine relation with respect to the adjusted forward process.
		More specifically,
		\begin{equation}\label{log-contract}
			\mathbb{E}\left[ X_{T+{\delta}}-{X_{T}}\big| \mathcal{F}_T\right]=c_1\int_{T}^{T+{\delta}}\left(1+b\left(E_{b,K}\ast 1\right)\left(T+{\delta}-s\right)\right)g_T\left(s\right) \dd s\le 0,\quad \mathbb{Q}-\text{a.s.},
		\end{equation}
		where $c_1=-(\frac{1}{2}+\int_{\mathbb{R}_{+}}(e^{-\Lambda z}-1+\Lambda z)\nu(\dd z))$.
	\end{theorem}
	\begin{proof}
		By \eqref{spot} and the martingale property of the local martingale part of the expression (see  \cite[Lemma $1$]{primo}), we have
		\begin{equation*}
			\mathbb{E}\left[ X_{T+{\delta}}-{X_{T}}\big| \mathcal{F}_T\right]= c_1 \int_{T}^{T+{\delta}}
			\mathbb{E}\left[ \sigma^2_s\big| \mathcal{F}_T\right]\dd s,\quad \mathbb{Q}-\text{a.s.}
		\end{equation*}
		Recalling that $\sigma^2\ge 0,\,\mathbb{Q}\otimes\dd t-$a.e., we infer that $\mathbb{E}[\sigma^2_s| \mathcal{F}_T]\ge0$ for a.e. $s>T,\,\mathbb{Q}-$a.s., hence the value of a log contract at time $T$ is nonpositive $\mathbb{Q}-$a.s. 
		
		By \eqref{canonical}, \eqref{forward}, the stochastic Fubini's theorem -- whose application is guaranteed by \cite[Lemma $1$]{primo} -- and a suitable change of variables, we infer that, $\mathbb{Q}-$a.s.
		\begin{align}\label{cont1}
			c_1^{-1}&\mathbb{E}\left[ X_{T+{\delta}}-{X_{T}}\big| \mathcal{F}_T\right]=\int_{0}^{T+{\delta}}f_0\left(s\right)\dd s-\int_{0}^{T}\sigma^2_s\,\dd s+\int_0^{T+{\delta}}\left(\int_{0}^{T}1_{\left\{r\le s\right\}}E_{b,K}\left(s-r\right)\dd \widetilde{Z}_r\right)\dd s
			\notag\\&=
			\int_{0}^{T+{\delta}}f_0\left(s\right)\dd s-\int_{0}^T\sigma^2_s\,\dd s+
			\int_{0}^{T}\left(E_{b,K}\ast 1\right)\left(T+{\delta}-r\right)\dd \widetilde{Z}_r\notag\\&
			=
			\int_{0}^{T+{\delta}}f_0\left(s\right)\dd s-\int_{0}^{T}\left(1+b\left(E_{b,K}\ast 1\right)\right)\left(T+{\delta}-s\right)\sigma^2_s\,\dd s+\int_{0}^{T}\left(E_{b,K}\ast 1\right)\left(T+{\delta}-r\right)\dd {Z}_r,
		\end{align}
		where $f_0=g_0-R_{-bK}\ast g_0$. Notice that $E_{b,K}\ast 1$ is the unique, continuous (nonnegative) solution of the linear Volterra equation $\chi=K\ast (1+b\chi)$. Then, another application of stochastic Fubini's theorem yields, $\mathbb{Q}-$a.s.,
		\begin{multline*}
			\int_{0}^{T}\left(E_{b,K}\ast 1\right)\left(T+{\delta}-r\right)\dd Z_r
			=
			\int_{0}^{T}\left(\int_{r}^{T+{\delta}}K\left(s-r\right)\left(1+b\left(E_{b,K}\ast 1\right)\left(T+{\delta}-s\right)\right)\dd s\right)\dd Z_r
			\\=
			\int_{0}^{T+{\delta}}\left(1+b\left(E_{b,K}\ast 1\right)\left(T+{\delta}-s\right)\right)\left(\int_{0}^{T}1_{\left\{r\le s\right\}}K\left(s-r\right)\dd Z_r\right)\dd s.
		\end{multline*}
		To conclude, we observe that by \cite[Theorem $2.2$ (\lowerRomannumeral{8}), Chapter $2$]{g}
		\[
		-\left(\left(R_{-bK}\ast g_0\right)\ast 1\right)\left(T+{\delta}\right)=b\left(\left( E_{b,K}\ast 1\right)\ast g_0 \right)\left(T+{\delta}\right),
		\]
		and plugging the previous two equalities in \eqref{cont1}, together with \eqref{variance}, \eqref{adjusted}, we obtain the relation in \eqref{log-contract}.
	\end{proof}
	We deduce the following corollary showing an affine relation between the square of the VIX index and the adjusted forward process. 
	\begin{corollary}\label{corollary-affine}
		The square of \emph{VIX}  satisfies an infinite dimensional affine relation with respect to the adjusted forward process. More specifically
		\begin{equation}\label{VIX2}
			\begin{array}{rcl}
				\emph{VIX}_T^2&=& \displaystyle -10^4\frac{2}{{\delta}}
				\mathbb{E}\left[ X_{T+{\delta}}-{X_{T}}\big| \mathcal{F}_T\right],\quad \mathbb{Q}-\text{a.s.} \\
				&=&\displaystyle -10^4\frac{2}{{\delta}}
				c_1\int_{T}^{T+{\delta}}\left(1+b\left(E_{b,K}\ast 1\right)\left(T+{\delta}-s\right)\right)g_T\left(s\right) \dd s,\quad \mathbb{Q}-\text{a.s.},
			\end{array}
		\end{equation}
		where $c_1=-(\frac{1}{2}+\int_{\mathbb{R}_{+}}(e^{-\Lambda z}-1+\Lambda z)\nu(\dd z))$.
	\end{corollary}
	
	\begin{rem}\label{Var_Swap}
		Our framework also allows us to obtain an explicit infinite dimensional affine relation between the variance swaps and the adjusted forward process. Specifically, the variance swap rate is 
		\begin{equation}\label{swaps}
		\frac{1}{\delta}	\mathbb{E}\left[\left[X,X\right]_{T+\delta}-\left[X,X\right]_T\big| \mathcal{F}_T\right]
		=
		\frac{c_2}{\delta}\int_{T}^{T+{\delta}}\left(1+b\left(E_{b,K}\ast 1\right)\left(T+{\delta}-s\right)\right)g_T\left(s\right) \dd s,\quad \mathbb{Q}-\text{a.s.},
		\end{equation}
		where $c_2=1+\Lambda^2\int_{\mathbb{R}_{+}}|z|^2\nu(\dd z)$.
		Note that for $\Lambda=0$ we have $c_2=-2c_1$, hence in this case  log contracts and variance swaps coincide up to the factor $-2/\delta$ (see \eqref{log-contract}-\eqref{swaps}). Therefore, when there are no jumps in the dynamics of the underlying, by \eqref{VIX2} we recover the fact that $\text{\emph{VIX}}^2$ is a variance swap. Moreover, observe that the relation in \eqref{swaps} is an extension of \cite[Lemma $4.4$]{Kallsen} in the classical affine setting. We refer to \cite{ContKok13, more, MST20} for more details regarding the distinction between variance swaps and \emph{VIX}$^2$.	
		\end{rem}
	%We easily remark that a difference arises between $VIX^2$ and variance swaps coherent with
	%the analysis in literature, see \cite{ContKok13, more, MST20}.  

	We are now interested in finding the conditional Fourier-Laplace transform of $\text{VIX}^2_T$.
	Before addressing this question, we need some technical intermediate steps.
	We first recall the following functional space as defined in \cite{ee2}
	\begin{equation}\label{invar}
		\mathcal{G}_K=\left\{g\colon\mathbb{R}_+\to\mathbb{R}\text{ continuous : }g\left(0\right)\ge 0\text{ and }\Delta_{\varepsilon}g-\left(\Delta_{\varepsilon}K\ast L\right)\left(0\right)g-\dd \left(\Delta_{\varepsilon}K\ast L\right)\ast g\ge0,\,{\varepsilon}\ge0 \right\}.
	\end{equation}
	 
	\begin{lemma}
	 Suppose that $K$ satisfies Hypothesis \ref{c1}. Define the function  $h\colon\mathbb{R}_+ \to \mathbb{R}$ by
	\[
	h\left(t\right)
	=-10^4\frac{2}{{\delta}}c_1
	\left[1+b\left(E_{b,K}\ast 1\right)\left({\delta}-t\right)\right]
	1_{\left\{t\le {\delta}\right\}},\quad t\ge0.
	\]
	Then $h$ is a continuous nonnegative function on $[0,\delta)$ and  $t\mapsto\int_{\mathbb{R}_+}h(s)K(s+t)\dd s$ belongs to $\mathcal{G}_K$.
	\end{lemma}
	
	\begin{proof}
	The first step is to show that $1+b(E_{b,K}\ast1)\ge0$ in $\mathbb{R}_+$, which implies that $h$ is also nonnegative. This can be deduced from the fact that this function is the unique, continuous solution  in $\mathbb{R}_+$ of the Volterra equation $\chi=1+bK\ast\chi$, which is nonnegative by \cite[Theorem C.$1$]{ee}. Secondly, $h$ has compact support, and under Hypothesis \ref{c1} for every ${\varepsilon}\ge 0$ the function $\Delta_{\varepsilon}K\ast L$ is right-continuous nondecreasing in $\mathbb{R}_+$ and (see the proof of \cite[Lemma $2.6$]{sergio})
	\[
	\Delta_{\varepsilon}K=\left(\Delta_{\varepsilon}K\ast L\right)\left(0\right)K+\dd \left(\Delta_{\varepsilon}K\ast L\right)\ast K,\quad \dd t-\text{a.e. in }\mathbb{R}_+.
	\]
	As a consequence, for every $t\ge0$
	\begin{multline*}
		\Delta_{\varepsilon}K\left(s+t\right)=
		\left(\Delta_{\varepsilon}K\ast L\right)\left(0\right)K\left(s+t\right)+\left(\dd \left(\Delta_{\varepsilon}K\ast L\right)\ast K\right)\left(s+t\right)\\
		\ge
		\left(\Delta_{\varepsilon}K\ast L\right)\left(0\right)K\left(s+t\right)+\int_{0}^{t}K\left(s+t-u\right)\dd \left(\Delta_{\varepsilon} K\ast L\right)\left( u\right)
		,\quad \text{for a.e. }s\in\left[0,{\delta}\right].
	\end{multline*}
	This implies, by Tonelli's theorem, that $t\mapsto\int_{\mathbb{R}_+}h(s)K(s+t)\dd s$ belongs to $\mathcal{G}_K$.	
        \end{proof}
	 
	We now define, for every $w\in\mathbb{C}_-$, the function $h_w(t)= w \cdot h(t),\,t\ge0,$ and consider the Riccati-Volterra equation  
	\begin{equation}\label{chunk}
		\phi_w=\int_{0}^\infty h_w\left(s\right)K\left(s+\cdot\right)	\dd s+K\ast \left(G\left(\phi_w\left(\cdot\right)\right)\right),
	\end{equation}
	where 
	\begin{equation}\label{G}
		G\left(u\right)= bu+\frac{c}{2}u^2+\int_{\mathbb{R}_+}\left(e^{uz}-1-uz\right)\nu\left(\dd z\right),\quad u\in\mathbb{C}_-.
	\end{equation}
	
	\begin{lemma}
		Suppose that $K$ satisfies Hypothesis \ref{c1}. For every $w\in\mathbb{C}_-$, there exists a unique continuous solution $\phi_w\colon\mathbb{R}_+\to\mathbb{C}_-$ to \eqref{chunk}. Moreover, 
		\begin{equation}\label{comp_VIX}
		\mathfrak{R}\phi_w\left(t\right)\le \phi_{\mathfrak{R}w}\left(t\right),\quad t\ge 0.
		\end{equation}
	\end{lemma}
	\begin{proof}
	Having in mind \cite[Theorem C.$1$]{ee}, the existence of a global  solution of \eqref{chunk} can be deduced as in \cite[Theorem $10$]{primo}, whereas the uniqueness of such $\phi_w$ is obtained with a procedure analogous to the proof of Theorem \ref{gl_log}, see \emph{Step \upperRomannumeral{3}} with $\Lambda=0$ in Appendix \ref{ap_A}. Moreover, again by analogy with \cite[Theorem $10$\,(\lowerRomannumeral{2})]{primo}, the comparison result \eqref{comp_VIX} holds. 
	\end{proof}
	Before stating the theorem that provides the conditional Fourier-Laplace transform of $\text{VIX}^2_T$, we define
	\begin{equation}\label{Phi}
		\Phi_w\left(t,s\right)= h_w\left(s-t\right)1_{\left\{s\ge t\right\}}+G\left(\phi_w\left(t-s\right)\right)1_{\left\{s< t\right\}},\quad t,s\ge0.
	\end{equation}
	\begin{theorem}\label{t6V}
		Assume the same hypotheses as in Theorem \ref{t6}. Then, for every $w\in\mathbb{C}_-$, 
		\begin{equation}\label{final_V}
			\mathbb{E}\left[\exp\left\{w\cdot\emph{VIX}^2_T\right\}\Big|\,\mathcal{F}_t\right]=\exp\left\{\widetilde{U}_t\left(w,T\right)\right\},\quad \mathbb{Q}-\text{a.s., }t\in\left[0,T\right],
		\end{equation}
		where $\widetilde{U}_t(w,T)=\int_{t}^{\infty}\Phi_w(T,s)g_t(s)\dd s,\,t\in[0,T]$.
	\end{theorem}
	\begin{proof}
		Fix $w\in\mathbb{C}_-$. First of all, notice that by the definition of $h_w$ and \eqref{VIX2}
		\begin{multline}\label{f_value}
			\widetilde{U}_T\left(w,T\right)=\int_{T}^{\infty}h_w\left(s-T\right)g_T\left(s\right)\dd s=-10^4\frac{2}{{\delta}}c_1 w\int_{T}^{T+{\delta}}\left(1+b\left(E_{b,K}\ast 1\right)\left(T+{\delta}-s\right)\right)g_T\left(s\right)\dd s
			\\=w\cdot\text{VIX}_T^2,\quad \mathbb{Q}-\text{a.s.}
		\end{multline}
		We introduce the process
		\begin{equation*}%\label{prolong}
			%\bar{f}_T\left(s\right)=\begin{cases}
			%	\sigma^2_s,&s\in\left[0,T\right],\\
			%	\mathbb{E}\left[\sigma^2_s\big|\mathcal{F}_T\right],&s>T,
			%\end{cases}\qquad\qquad 
			\bar{g}_T\left(s\right)=\begin{cases}
				\sigma^2_s,&s\in\left[0,T\right],\\
				g_T\left(s\right),&s>T.
			\end{cases}
		\end{equation*}
		Note that by \eqref{variance} and \eqref{adjusted}, $\bar{g}_T(\cdot)$  is a joint measurable modification of $g_0+\int_{0}^{T}1_{\{r\le \cdot\}}K(\cdot-r)\dd Z_r$. For every $t\in[0,T]$, the stochastic Fubini's theorem, \eqref{chunk}, \eqref{Phi}, and suitable changes of variables, yield
		\begin{align}
			&\int_{0}^{\infty}\Phi_w\left({T},s\right)\left(\bar{g}_t\left(s\right)-g_0\left(s\right)\right)\dd s=
			\int_{0}^{\infty}\Phi_w\left({T},s\right)\left(\int_{0}^{t}1_{\left\{u\le s\right\}}K\left(s-u\right)\dd Z_u\right)\dd s\notag\\
			&\qquad=
			\int_{0}^{t}\left(\int_{0}^{\infty}h_w\left(s\right)K\left(s+{T}-u\right)\dd s
			+\int_{0}^{{T}-u}K\left(s\right)G\left(\phi_w\left({T}-u-s\right)\right)\dd s\right)\dd Z_u\notag\\
			&\qquad=\int_{0}^{t}\phi_w\left({T}-u\right)\dd Z_u,\quad \mathbb{Q}-\text{a.s.}
			\label{p1}
		\end{align}
		Moreover, by \eqref{Phi}, the following equality holds:
		\begin{multline}\label{p2}
			\int_{0}^{t}\Phi_w\left({T},s\right)\sigma^2_s\,\dd s=\int_{0}^{t}h_w\left(s-{T}\right)1_{\left\{s\ge{T}\right\}}\sigma^2_s\,\dd s+
			\int_{0}^{t}G\left(\phi_w\left({T}-s\right)\right)\sigma^2_s\,\dd s\\=\int_{0}^{t}G\left(\phi_w\left({T}-s\right)\right)\sigma^2_s\,\dd s.
		\end{multline}
		Recalling the definition of $\widetilde{U}_t(w,T)$, we combine \eqref{p1} and \eqref{p2} to write
		\begin{multline}\label{cadlag}
			\widetilde{U}_t\left(w,T\right)=\int_{t}^{\infty}\Phi_w\left({T},s\right)g_0\left(s\right)\dd s+\int_{0}^{\infty}\Phi_w\left({T},s\right)\left(\bar{g}_t\left(s\right)-g_0\left(s\right)\right)\dd s-\int_{0}^{t}\Phi_w\left({T},s\right)\left(\sigma^2_s-g_0\left(s\right)\right)\dd s\\
			=
			\int_{0}^{\infty}\Phi_w\left({T},s\right)g_0\left(s\right)\dd s+\int_{0}^{t}\phi_w\left({T}-u\right)\dd Z_u-\int_{0}^{t}G\left(\phi_w\left({T}-s\right)\right)\sigma^2_s\,\dd s,\quad \mathbb{Q}-\text{a.s.}
		\end{multline}
		In the sequel we denote by $U(w,T)=(U_t(w,T))_{t\in[0,T]}$ the c\`adl\`ag process defined by the rightmost side of \eqref{cadlag}. An application of {It\^o's formula} together with \eqref{G} shows that $E(w,T)=(\exp\{U_t(w,T)\})_{t\in[0,T]}$ is a local martingale, namely $E(w,T)=\exp\{\int_{0}^{\infty}\Phi_w(T-s)g_0(s)\dd s\}\mathcal{E}(\widetilde{N}(w,T))$, where $\mathcal{E}$ denotes the Dol\'eans-Dade exponential and  $\widetilde{N}(w,T)=(\widetilde{N}_t(w,T))_{t\in[0,T]}$ is defined by
		\[
		\dd \widetilde{N}_t\left(w,T\right)=\sqrt{c}\,\phi_w\left(T-t\right)\sigma_t\,\dd W_{2,t}+\int_{\mathbb{R}_{+}}\left(e^{\phi_w\left(T-t\right)z}-1\right)\tilde{\mu}\left(\dd t, \dd z\right),\quad \widetilde{N}_0\left(w,T\right)=0.
		\]
		
		If $E(w,T)$ is a true martingale, then \eqref{final_V} follows from \eqref{f_value} and \eqref{cadlag}. As in the proof of Theorem \ref{t6}, we search for an expression of $U(w,T)$ which is affine on the past trajectory of $\sigma^2$. However, we cannot directly invoke \cite[Theorem $7$]{primo} due to the different structure of the Riccati-Volterra equation in \eqref{chunk} and of the process $U(w,T)$ itself. Fortunately, we can adapt the procedure in the proof of \cite[Theorem $7$]{primo}. Specifically, thanks to the local boundedness of $\Phi_w(T,\cdot)$ (see~\eqref{Phi}), $\mathbb{Q}-$a.s.,
		\[
		U_t\left(T,w\right)=\int_{t}^{T+{\delta}} \Phi_w\left({T},s\right)g_0\left(s\right)\dd s+
		\phi_w\left({T}-t\right)Z_t
		+
		\left(\pi_{{T}+{\delta}-t}\ast\left(\sigma^2-g_0\right)\right)\left(t\right),\quad \text{for a.e. }t\in\left(0,T\right).
		\]
		Here the functions $$\pi_{{T}+{\delta}-t}(u)=\int_{0}^{{T}+{\delta}-t}\Phi_w({T},{T}+{\delta}-s)((\Delta_{{T}+{\delta}-t-s}K)'\ast L)(u)\dd s,\quad t\in(0,T),$$ are well defined for almost every $u\in\mathbb{R}_+$ and belong to $L^1_\text{loc}(\mathbb{R}_+).$ At this point, for every $t\in(0,T)$ we introduce the locally absolutely continuous function
		\begin{multline*}
			\widetilde{\Pi}_{{T}+{\delta}-t}\left(u\right)=\int_{0}^{u}\pi_{{T}+{\delta}-t}\left(s\right)\dd s+\phi_w\left({T}-t\right)L\left(\left[0,u\right]\right)
			\\=\int_{0}^{{T}+{\delta}-t}\Phi_w\left({T},{T}+{\delta}-s\right)\left(\left(\Delta_{{T}+{\delta}-t-s}K\right)\ast L\right)\left(u\right) \dd s,\quad u\ge0,
		\end{multline*}
		where the second equality is due to \eqref{chunk} and a suitable change of variables. Therefore, also recalling \eqref{variance}, 	the previous formula for $U(w,T)$ can be rewritten as, $\mathbb{Q}-$a.s., for a.e. $t\in(0,{T})$,
		\begin{equation*}
			U_t\left(w,T\right)=\int_{t}^{{T}+{\delta}} \Phi_w\left({T},s\right)g_0\left(s\right)\dd s
			+
			\left(\dd \widetilde{\Pi}_{{T}+{\delta}-t}\ast\left(\sigma^2-g_0\right)\right)\left(t\right)\\+\phi_w\left(T-t\right)L\left(\left\{0\right\}\right)\left(\sigma^2-g_0\right)\left(t\right),
		\end{equation*}
		which is an affine expression in terms of the past trajectories of $\sigma^2$. Now by Lemma \ref{prep} the real-valued process $E(\mathfrak{R}w,T)=(\exp\{U_t(\mathfrak{R}w,T)\})_{t\in[0,T]}$ is a true martingale. Thus, thanks to \eqref{comp_VIX}, we can parallel the comparison argument in the proof of  \cite[Theorem $11$]{primo} to deduce that 
		\[
		\left|\exp\left\{U_t\left(w,T\right)\right\}\right|=\exp\left\{\mathfrak{R}U_t\left(w,T\right)\right\}\le C\exp\left\{U_t\left(\mathfrak{R}w,T\right)\right\},\quad t\in\left[0,T\right],\,\mathbb{Q}-\text{a.s.},
		\]
		for some constant $C>0$. An application of \cite[Lemma $1.4$]{J} completes the proof.
	\end{proof}
	\begin{comment}
		\begin{itemize}
			\item Let $\bar{\psi}_w(\cdot;\bar{T})\colon\mathbb{R}_+\to\mathbb{R}_-$ be a real, nonpositive, continuous solution of the Riccati-Volterra equation
			\[
			\chi=\int_{\mathbb{R}_+}\mathfrak{R}h_w\left(s;\bar{T}\right)K\left(s+\cdot\right)ds+K\ast \left(F\left(\chi\right)\right).
			\] 
			The argument in Appendix B proves that $\mathfrak{R}\psi_w\le \bar{\psi}_w$ in $\mathbb{R}_+.$ Moreover, for every $t\in[0,\bar{t}]$, explicit computations show that 
			\begin{multline*}
				\tilde{\bar{\Pi}}_{\bar{t}+\bar{T}-t}\left(r\right)-\mathfrak{R}\left(\tilde{\Pi}_{\bar{t}+\bar{T}-t}\left(r\right)\right)=
				-\int_{\left(0,\bar{t}-t\right]}\left[\bar{\psi}_w\left(\bar{t}-t-u\right)-\mathfrak{R}\psi_w\left(\bar{t}-t-u\right)\right]L\left(r+ds\right)
				\\+
				\int_{0}^{\bar{t}-t}\left(F\left(\bar{\psi}_w\left(s\right)\right)-\mathfrak{R} F\left(\psi_w\left(s\right)\right)\right)ds,\quad r\ge0,
			\end{multline*}
			hence it is an increasing function of $r$. Now one can follow the procedure in Theorem $8$ to conclude that $\exp\{V^w\}$ is a true martingale.
		\end{itemize}
	\end{comment}
	
	\subsection{VIX put options and futures prices}
	Theorem \ref{t6V} provides a semi-explicit formula for the Fourier-Laplace transform $\lambda_T$ of $\text{VIX}_T^2$ in $\mathbb{C}_-$, namely
	\begin{multline}\label{LP_V}
		\lambda_T\left(w\right)
		=\mathbb{E}\left[\exp\left\{w\cdot\text{VIX}^2_T\right\}\right]=
		\exp\left\{\int_{0}^\infty \Phi_w\left(T,s\right)g_0\left(s\right)\dd s\right\}
		\\=
		\exp\left\{ \int_{0}^{{\delta}}h_w\left(s\right)g_0\left(s+T\right)\dd s+ \left(g_0\ast G\left(\phi_w\left(\cdot\right)\right)\right)\left(T\right)
		\right\}
		,\quad w\in\mathbb{C}_-.
	\end{multline}
	This allows us to price put options written on $\text{VIX}$ with the Fourier-inversion technique for the bilateral Laplace transform shown in \cite{chic}. More specifically, for a log strike $k\in\mathbb{R}$, the payoff function of such options defined on the whole real line  is  $w(x)=(e^k-\sqrt{x^+})^+,\,x\in\mathbb{R}$, where $x^+$ represents $\text{VIX}^2$. Then, denoting by $P(k,T)$ the price of a put option with maturity $T$ (and log strike $k$) we have (cfr. \cite[Equations $(7.6)$-$(7.8)$]{chic}) 
	\begin{align}\label{pr_VIX}
		P\left(k,T\right)&=\mathbb{E}\left[\left(e^k-\text{VIX}_T\right)^+\right]=
		-\frac{1}{4\sqrt{\pi} i}\int_{z_r-i\infty}^{z_r+i\infty}\frac{\text{erf}\left(e^k\sqrt{z}\right)}{z^{3/2}}\lambda_T\left(z\right)\dd z\notag\\&
		=-\frac{1}{4\sqrt{\pi}}\int_{\mathbb{R}}\mathfrak{R}\left[\frac{\text{erf}\left(e^k\sqrt{z_r+iu}\right)}{\left(z_r+iu\right)^{3/2}}\lambda_T\left(z_r+iu\right)\right]\dd u\notag
		\\
		&=	-\frac{1}{2\sqrt{\pi}}\int_{\mathbb{R}_+}\mathfrak{R}\left[\frac{\text{erf}\left(e^k\sqrt{z_r+iu}\right)}{\left(z_r+iu\right)^{3/2}}\lambda_T\left(z_r+iu\right)\right]\dd u,\quad z_r<0.
	\end{align}
	{\color{black} Here $\text{erf}$ represents the error function $\text{erf}\,z=\frac{2}{\sqrt{\pi}}\int_{0}^{z}e^{-t^2}\dd t,\,z\in\mathbb{C}$}, and for $z\in\mathbb{C}$ and $a\ge0$, we consider the power $z^a=\Lambda^a e^{ia\theta},$ where $z=\Lambda e^{i\theta}$ with $\Lambda\ge0,\,\theta\in(-\pi,\pi]$. In particular, we write $\sqrt{z}=z^{1/2}$. The last equality in \eqref{pr_VIX} is due to the fact that the integrand is even. Indeed, this follows from the well-known symmetry relation $\text{erf}\,\bar{z}=\widebar{{\text{erf}\,z}},\,z\in\mathbb{C}$, as well as the identities (for $u\neq 0$)
	\[
	\mathfrak{R}\left(\sqrt{z_r+iu}\right)=\sqrt{\frac{z_r+\sqrt{z_r^2+u^2}}{2}},\qquad \mathfrak{Im}\left(\sqrt{z_r+iu}\right)=\text{sgn}\left(u\right)\sqrt{\frac{-z_r+\sqrt{z_r^2+u^2}}{2}}.
	\]
	
	Moreover, we can use $\lambda_T$ to determine $\mathbb{E}[\text{VIX}_T],$ i.e., the futures price of VIX at time $T$. In order to do this, notice that for every $x\ge0$ the function $(\sqrt{\pi s})^{-1}(e^{-xs}-1)+\sqrt{x}\,\text{erf}(\sqrt{sx}),\,s> 0$, is an antiderivative of $(2\sqrt{\pi})^{-1}(1-e^{-xs})s^{-3/2},\,s>0$. From this relation we deduce the following integral representation for the square-root function 
	\[
	\sqrt{x^+}=\frac{1}{2\sqrt{\pi}}\int_0^\infty \frac{1-e^{-sx^+}}{s^{\frac{3}{2}}}\,\dd s,\quad x\in\mathbb{R}.
	\]
	An application of Tonelli's theorem yields
	\begin{equation}\label{futures}
		\mathbb{E}\left[\text{VIX}_T\right]=\frac{1}{2\sqrt{\pi}}\int_{0}^{\infty}\frac{1-\lambda_T\left(-s\right)}{s^{\frac{3}{2}}}\,\dd s.
	\end{equation}

	\section{Numerical approximation of the model}\label{sec_num}
	According to the formulae in \eqref{pr_SPX}-\eqref{put_SPX} and \eqref{pr_VIX}, in order to price options on $S$ and VIX with maturity $T$, one needs to compute $\Psi^{X_T}(w_1)$ and $\lambda_T(w_2)$, where $w_1,\,w_2$ belong to appropriate regions of $\mathbb{C}$. In addition, the values  $\lambda_T(-s),\,s\ge0,$ are also necessary to determine the futures price of VIX at time $T$. Consequently, looking at the expressions of these Fourier-Laplace transforms in \eqref{LP_X} and \eqref{LP_V}, the solutions of the Riccati-Volterra equations \eqref{RV_log} and \eqref{chunk}, i.e., $\psi_{w_1},\,\phi_{w_2}$ and $\phi_{-s}$, have to be approximated on the interval $[0,T]$. Among the available numerical methods to approximate them we choose the multi-factor scheme suggested in \cite{ee}. Another possibility would have been the Adams scheme \cite{Ada1,Ada2} or hybrid schemes as in \cite{CGP}. The multi-factor scheme consists in approximating the kernel $K$ with a weighted sum of exponentials, namely with functions $K_n,\,n\in\mathbb{N},$ of the form	
	\begin{equation}\label{K_n}
		K_n\left(t\right)=\sum_{j=1}^{n}m_{j,n}e^{-x_{j,n}t},\quad t\ge0,
	\end{equation}
	where $m_{j,n},x_{j,n}>0,\,j=1,\dots,n$. In what follows, we write  $\mathbf{m}=\{m_{j,n}|\, j=1,\dots,n,\,n\in\mathbb{N}\}$ and $\mathbf{x}=\{x_{j,n}|\, j=1,\dots,n,\,n\in\mathbb{N}\}$.  Notice that $K_n,\,n\in\mathbb{N},$ is completely monotone on $(0,\infty)$,
	{\color{black} meaning that it is nonnegative and infinitely differentiable on this interval,  with nonpositive [resp., nonnegative] odd [resp., even] $k-$derivative, $k\in\mathbb{N}$.}
	% \footnote{A function $f$ is called completely monotone on $(0,\infty)$ if it is
	%	infinitely differentiable there with $(-1)^kf^{(k)}(t) \ge 0$ for all $t > 0$ and $k = 0, 1, \ldots$.}.
	%are given by the expression in  \cite[Equation $\left(3.6\right)$]{ee} and  depend on the choice of a non--negative sequence $\left(\rho_{j,n}\right)_{j=0,\dots,n}$. 
	More details about this approximation and the idea behind it can be found in Remark \ref{approx_ker} below and in the references therein.
	
	Given $n\in\mathbb{N}$ and $w\in\mathbb{C}$ such that $\mathfrak{R}w\in[0,1]$, we now introduce the Riccati-Volterra equation
	\begin{equation}\label{RV_log_n}
		\psi_{w,n}\left(t\right)=\int_{0}^{t}K_n\left(t-s\right) F\left(w, \psi_{w,n}\left(s\right)\right)\dd s=\left(K_n\ast F\left(w,\psi_{w,n}\left(\cdot\right)\right)\right)\left(t\right),\quad t\ge0.
	\end{equation}
	Note that the existence and uniqueness of $\psi_{w,n}$ is guaranteed by Theorem \ref{gl_log}\,(\lowerRomannumeral{1}), because $K_n$ satisfies Hypothesis \ref{c1}. The advantage in considering \eqref{RV_log_n} instead of \eqref{RV_log} is that its solution $\psi_{w,n}$ can be obtained by numerically solving a system of integral equations with standard methods. More precisely, $\psi_{w,n}(t)=\sum_{j=1}^{n}m_{j,n}\psi^{(j)}_{w,n}(t)$ for every $t\ge0$, where
	\[
	\psi^{\left(j\right)}_{w,n}\left(t\right)= e^{-x_{j,n}t}\int_{0}^{t}e^{x_{j,n}s}F\left(w,\sum_{k=1}^{n}m_{k,n}\psi^{\left(k\right)}_{w,n}\left(s\right)\right)\dd s,\quad j=1,\dots,n.
	\]
	Analogously, for every $n\in\mathbb{N}$ and $w\in\mathbb{C}_-$, we consider the Riccati-Volterra equation
	\begin{equation}\label{VIX_n}
		\phi_{w,n}\left(t\right)=\int_{0}^\infty h_w\left(s\right)K_n\left(s+t\right)	\dd s+\left(K_n\ast \left(G\left(\phi_{w,n}\left(\cdot\right)\right)\right)\right)\left(t\right),\quad t\ge0.
	\end{equation}
	We have that $\phi_{w,n}(t)=\sum_{j=1}^{n}m_{j,n}\phi^{(j)}_{w,n}(t)$, $t\ge0$, with
	\[
	\phi^{\left(j\right)}_{w,n}\left(t\right)=e^{-x_{j,n}t}\left(\int_{0}^{\infty}h_w\left(s\right)e^{-x_{j,n}s}\dd s
	+\int_{0}^{t}e^{x_{j,n}s}G\left(\sum_{k=1}^nm_{k,n}\phi_{w,n}^{\left(k\right)}\left(s\right)\right)\dd s\right),\quad j=1,\dots,n.
	\]
	The following theorem offers an estimate on the uniform distance on $[0,T]$ between $\psi_w$ and $\psi_{w,n}$, as well as  between $\phi_w$ and $\phi_{w,n}$.  In the former case, it generalizes \cite[Theorem $4.1$]{ee} to our framework with jumps. Its proof, which we postpone to Appendix \ref{ap_C}, relies on results related to Riccati-Volterra equations which are proved in Appendix \ref{ap_A}.  
	\begin{comment}
		{\begin{lemma}\label{pr}
				Take $p,\,q\ge0$. Then $E_{p+q,K}(t)\le e^{qt}E_{p,K}$ for a.e. $t\ge0$.
			\end{lemma}
			\begin{proof}
				Fix $p,\,q\ge0$; observe that $E_{p+q,K}$ solves the linear Volterra equation $\chi=K+(p+q)K\ast \chi$ in $\mathbb{R}_+$. Then denoting by $\widetilde{E}$ the unique $L^2_\text{loc}(\mathbb{R}_+)$ function satisfying $\chi(t)=e^{qt}K(t)+pe^{qt}(K\ast \chi)(t)$ for a.e. $t\ge0$, we can claim that $E_{p+q,K}\le \widetilde{E}$ a.e. in $\mathbb{R}_+$. Indeed, since by Hypothesis \ref{c1} $|K(t)|\le e^{qt}K(t)$ and $(p+q)K(t)\le pe^{qt}K(t)$ for $t>0$, this assertion follows from an application of 
		\end{proof}}
	\end{comment}
	\begin{theorem}\label{stime}
		Assume that $K$ satisfies Hypothesis \ref{c1}. Let $T>0$ and denote by $E_{\lambda,n}$ the canonical resolvent of $K_n$ with parameter $\lambda\in\mathbb{R}$, $n\in\mathbb{N}$. 
		\begin{enumerate}[label=$\left(\emph{\roman*}\right)$]
			\item\label{t11} Suppose that $\int_{0}^{T}|E_{b+\rho^+\sqrt{c},n}(s)|\dd s\le \widetilde{C}$ for every $n\in\mathbb{N}$, where $\widetilde{C}=\widetilde{C}(\rho,b, \mathbf{m},\mathbf{x},T)>0$. Then there exists a constant $C=C(\rho,b,c,\Lambda,\nu, \mathbf{m}, \mathbf{x}, T)>0$ such that, for every $w\in\mathbb{C}$ with $\mathfrak{R}w\in[0,1]$ and $n\in\mathbb{N}$, 
			\begin{multline}\label{bou_log}
				\sup_{t\in\left[0,T\right]}\left|\psi_w\left(t\right)-\psi_{w,n}\left(t\right)\right|\le C\left(1+\left|\mathfrak{Im}w\right|^6\right){\int_{0}^{T}E_{C\left(1+\left|\mathfrak{Im}w\right|^2\right),K}\left(s\right)\dd s}
				\\\times\int_{0}^{T}\left|K_n\left(s\right)-K\left(s\right)\right|\dd s.
			\end{multline}
			In addition, if $b<0$ and $\rho<0$ then the constant $C$ does not depend on $\mathbf{m}$ or $\mathbf{x}$, and the dependence on $T$ is via $\norm{K}_{L^1(\left[0,T\right])}$.
			\item \label{t12}
			Suppose that $\int_{0}^{T\vee {\delta}}|E_{b^+\!,n}(s)|\dd s\le \widetilde{C}$ for every $n\in\mathbb{N}$, where $\widetilde{C}=\widetilde{C}(b, \mathbf{m},\mathbf{x},T,{\delta})>0$. Then there exists a constant $C=C(b,c,\Lambda,\nu, \mathbf{m}, \mathbf{x}, T,{\delta})>0$ such that, for every $w\in\mathbb{C_-}$ and $n\in\mathbb{N}$, 
			\begin{equation}\label{bou_vix}
				\sup_{t\in\left[0,T\right]}\left|\phi_w\left(t\right)-\phi_{w,n}\left(t\right)\right|\le C\left(1+\left|w\right|^6\right){\int_{0}^{T}E_{C\left(1+\left|w\right|^2\right),K}\left(s\right)\dd s}\int_{0}^{T\vee{\delta}}\left|K_n\left(s\right)-K\left(s\right)\right|\dd s.
			\end{equation}
		\end{enumerate}
	\end{theorem}
	
	\begin{rem}\label{approx_ker}
		When the kernel $K$ is completely monotone, a standard way to determine $\mathbf{m}$ and $\mathbf{x}$ in \eqref{K_n} relies on the Bernstein-Widder theorem (see, e.g., \cite[Theorem $2.5$, Chapter~$5$]{g}), according to which there exists a nonnegative measure $\mu$ on $\mathbb{R}_+$ such that $K(t)=\int_{\mathbb{R}_{+}}e^{-xt}\mu(\dd x),\,t>0$. Approximating $\mu$ with a weighted sum of Dirac measures gives $K_n$. More specifically, for a fixed $n\in\mathbb{N}$ it is customary to take a strictly increasing sequence of nonnegative numbers $(\rho_{j,n})_{j=0,\dots,n}$, and then choose for every $j=1,\dots,n$,
		\[
		m_{j,n}=\int_{\rho_{j-1,n}}^{\rho_{j,n}}\mu\left(\dd y\right),\qquad x_{j,n}=m_{j,n}^{-1}\int_{\rho_{j-1,n}}^{\rho_{j,n}}y\,\mu\left(\dd y\right).
		\]
		We mention that in some instances (most notably when $K$ is the fractional kernel, see, e.g., \cite[Example $2.6$]{american_sergio}) it is possible to show the convergence $K_n\to K$ in $L^2_{\emph{loc}}(\mathbb{R}_+)$. Thanks to \cite[Theorem $3.1$, Chapter $2$]{g}, this ensures the validity of the hypotheses required in both points of Theorem \ref{stime}, and therefore the convergence of the multi-factor scheme.
	\end{rem}
	
\section{Calibration}\label{sec_cal}

We have shown that efficient Fourier-based methods can be applied to the rough Hawkes Heston model in order to price options on the underlying prices and their volatility index. 
Based on these techniques, in this section we calibrate a parsimonious specification of the rough Hawkes Heston model to S\&P 500 and VIX options data on May 19th, 2017. 
This is the same data set as in \cite{Gatheral}. 
Here the process $X$ represents the logarithm of the ratio between the spot underlying price and the forward at initial time. 
For our parametrization, we choose -- as in the rough volatility models -- a power kernel of the form $t^{\alpha-1}/\Gamma(\alpha)$, $\alpha\in(1/2,1]$. 
For the law of the jumps, and to keep the number of parameters low, we choose an exponential distribution with rate 1, $\nu(\dd z)=\exp(-z)\,\dd z$. 
Our parsimonious specification of the model has therefore -- other than the two parameters related to the initial variance curve $(\beta,\sigma_0^2)$ -- five evolution-related parameters $(\alpha,  \rho, b, c, \Lambda)$. 
Like in \cite{Gatheral}, we concentrate on short maturities for which, as pointed out in \cite{Guyon2}, ``VIX derivatives are most liquid and the joint calibration is most difficult.''
The resulting calibrated parameters are reported in Table \ref{Table:CalibratedParameters}. 

\begin{table}[h!]
    \centering
    \begin{tabular}{|c|c|c|c|c||c|c|}
    \hline
        $\alpha$ & $ \rho$  &   $b$ & $c$ & $\Lambda$  &$\beta$  & $\sigma^2_0$\\ \hline
        0.506 & -0.737  & -2.008 &  0.156 & 0.242  & 0.048 & 0.007 \\ \hline
    \end{tabular}
    \caption{Calibrated parameters}
    \label{Table:CalibratedParameters}
\end{table}

We observe that the value of $\alpha$ is very close to its lower bound limit $0.5$. 
This is coherent with previous estimates in the  rough volatility literature, see for instance \cite{Alos-S, Bayer, bennedsen_pakkanen,EEFR, Fuka, Gath1, Gatheral}. %\cite{Fuka, Gath1, Gatheral}. 
%Our analysis confirm then the consistency of this crucial estimate. 
The estimation of the correlation parameter $\rho$ is also in line with empirical estimates, e.g. \cite{Cont01}, and what is commonly known as the {\it leverage effect} \cite{CS15,EEFR, MT22}. 
We notice that for the joint calibration we can keep the vol-of-vol parameter $c$ small because an important part of the volatility fluctuation is captured by the {\it self-exciting} jumps controlled by the parameters $\alpha$ and $\Lambda$.
This responds to the issue, raised in \cite{Guyon2}, that  ``very large negative skew of short-term SPX options, which in continuous models implies a very large volatility of volatility, seems inconsistent with the comparatively low levels of VIX implied volatilities.''

 The calibrated implied volatility smiles for the S\&P 500 and VIX options are plotted in Figures \ref{graph-calibration-SP} and \ref{graph-calibration-VIX}, respectively. 
We zoom the calibration of the S\&P 500 options around the money in Figure \ref{graph-calibration-SP-zoom}. 
These graphs show that the model fits remarkably well both S\&P 500 and VIX implied volatilities. 
The shape of the smile around the money for S\&P 500 options is well-captured and the distance to the bid-ask corridor -- across the maturities -- is at most of one bid-ask spread. 
For the two shortest maturities most of the model implied volatilities around the money are actually inside the bid-ask corridor.
The fit is not perfect for very negative log-moneyness.
This is also seen -- possibly to a less extent -- in the quadratic rough Heston model \cite{Gatheral}.
We conjecture that, at the cost of increasing the complexity of the model, even better results could be obtained if we replace the exponential law for the jumps by a law with Pareto tails as suggested in \cite{Cont01, JMSZ} and the references therein. 
Regarding the VIX implied volatilities, we observe that -- even for options deep out-of the-money -- the model implied volatilities stay almost systematically within the bid-ask corridor, whether it is calculated using call or for put options.

\section{Sensitivities of the implied volatilities}\label{sec_sensi}

In this section we study the sensitivity of the implied volatilities of S\&P 500 and VIX options to the parameters of the rough Hawkes Heston model. 
Starting from the calibrated parameters presented in Table \ref{Table:CalibratedParameters}, we analyze the impact of a change in the evolution-related parameters $(\alpha,  \rho, b, c, \Lambda)$ and the initial curve parameters $(\beta,\sigma_0^2)$ on the implied volatilities for the shortest maturity, and the shortest and longest maturities, respectively.

We begin with the sensitivity with respect to the parameter $\alpha\in (0.5,1]$, which as we will see plays a crucial role in our model. We can observe in Figure \ref{fig:sensitivity-alpha} - as is the case for other rough volatility models -- that modifications of the parameter $\alpha$ change the ATM skew of the implied volatility of S\&P 500 options. The right convexity and ATM skew around the money can only be obtained for very low values of the parameter $\alpha$, confirming the findings in the rough volatility literature. To elucidate the influence of the parameter $\alpha$ on the ATM skews, we plot in Figure \ref{graph-skew-decay} the log-log plots of ATM skews as a function of maturity, for the calibrated parameters and different values of $\alpha$. We observe that a perfect power decay is captured by $\alpha=0.506$, but not by higher values of $\alpha$. For $\alpha=0.506$, the linear fit is almost perfect with a $-0.597$ power decay and an unquestionable coefficient of determination $R^2=0.99905$. These findings are coherent with the results in the rough volatility literature, e.g. \cite{Bayer, Gath1}, indicating a power law for the ATM skew as a function of maturity given approximately by $T^{-\frac12}$. For other values of $\alpha$, the linear fit is also observed for the shortest maturities. We plot in Figure \ref{graph-skew-decay}, the estimated power decay for the short maturities as a function of $\alpha$. This plot shows that the relationship between the power decay and $\alpha$ is approximately linear.

More importantly, within the joint calibration framework, the parameter $\alpha$ has a big impact on the level and shape of implied volatilities of VIX options. This is confirmed by Figure \ref{fig:sensitivity-alpha}. In particular, the difference in level between the implied volatilities of VIX options for $\alpha =0.506$ and $\alpha=0.6$ is similar to the one between $\alpha =0.6$ and $\alpha=0.9$. As $\alpha$ decreases the implied volatilities shift downwards. This feature is fundamental to bring down the VIX implied volatilities maintaining the correct skew for SPX implied volatilities, explaining therefore the shift mentioned in \cite{Guyon, Guyon2}. We ratify therefore -- within the affine framework -- the relevance of rough non-Markovian volatility to jointly calibrate SPX and VIX smiles.

We now analyze the dependency of the implied volatilities with respect to the other parameters. 
Figure \ref{fig:sensitivity-others} shows the sensitivities with respect to the evolution-related parameters $(b,c, \rho, \Lambda)$.
We notice that -- unless we zoom around the money -- the sensitivity of the SPX smiles with respect to $(b,c,\Lambda)$ is relatively small. 
The main effect of an increment in the reverting speed $-b$ is a shift slightly downwards of the SPX implied volatility and a more pronounced upward shift and a reduction of the concavity on the VIX implied volatility.
The impact of the volatility of volatility $c$ is similar for SPX options, with a slight change of concavity, and a more pronounced and less symmetric effect on the level and concavity of implied volatility of VIX options. 
As usual, the correlation parameter $\rho$ plays a big role by moving the minimum value to the left ($\rho<0$) or to the right ($\rho>0$).
Obviously, the VIX smiles do not depend on the correlation $\rho$.
The effect of the (jump) leverage $\Lambda$ is relatively small on SPX implied volatilities but fundamental on the VIX implied volatilities. 
For SPX implied volatilities, the impact of $\Lambda$ could be reduced to a rotation with the money as pivot.
The parameter $\Lambda$ also controls the level of VIX implied volatility out-of-the-money. As $\Lambda$ increases this level goes down, achieving the correct shift for the calibrated parameter. 
This effect is similar to the one observed for the vol-of-vol $c$, but the sensitivity is larger, and it allows us to keep a low value of $c$ for the joint calibration.
This explains, the importance in our model of self-exciting jumps in opposite directions for the underlying and volatility.

We now turn to the parameters $(\beta,\sigma_0^2)$ of the initial curve $g_0(t)=\sigma_0^2+\beta\int_0^t K(s)\,\dd s$, $t\ge 0$.
Figure \ref{fig:sensitivity-spot-curve} shows the SPX and VIX implied volatility sensitivities for the shortest and longest maturity.
The impact of both parameters is similar for SPX and VIX options.
When $\sigma_0^2$ or $\beta$ increase the SPX implied volatilities move up and to the right, while the VIX implied volatilities move down and the concavity increases.

\section{Conclusion}\label{sec_conclusion}

We develop and study a new stochastic volatility model named the rough Hawkes Heston model. It is a tractable affine Volterra model with rough volatility and volatility jumps that cluster and that have the opposite direction but occur at the same time as the jumps of the underlying prices. This model shares many features with other existing models, mainly the Heston \cite{Heston},  Barndorff-Nielsen and Shephard \cite{BarShe}, and rough Heston \cite{rh2} models. It takes advantage of the low regularity and memory features of rough volatility models, the large fluctuation of jumps, the clusters of Hawkes processes and the explicit Fourier-Laplace transform of the affine setup. By combining the modeling advantages of these approaches, it is able to better capture the joint dynamics of underlying prices and their volatility index in a tractable fashion. The addition of a singular kernel in the dynamics of the volatility, together with jumps, incorporates not only the rough volatility feature but also a jump-clustering component. The presence of common jumps in the underlying and the volatility in opposite directions is coherent with previous studies such as \cite{Todorov}. Moreover, the introduction of jumps that cluster -- as in  \cite{BBSS} --  is in accordance with empirical findings, e.g. \cite{Cont01, Cont11}. Similar to \cite{BarShe,rh2,Heston}, the rough Hawkes Heston model is parsimonious with only five evolution-related parameters, and it belongs to the class of affine Volterra models \cite{sergio,primo}, which allows efficient Fourier-based techniques for pricing.

The parameter $\alpha$ describing the power kernel in the volatility dynamics controls -- as in the rough Heston model -- the underlying implied volatilities ATM skews for short maturities. Our calibration example indicates that this value is close to 0.5 which agrees with previous estimates in the rough volatility literature \cite{Bayer, Gath1}. This is not, however, the only role played by the parameter $\alpha$ in our setup because the power kernel also introduces a jump-clustering feature to the model. As a consequence, the parameter $\alpha$ plays a crucial role in controlling the level of VIX implied volatilities. Together with the jump-leverage parameter $\Lambda$, the power kernel allows us to bring down the VIX implied volatilities maintaining the correct skew for SPX implied volatilities, consequently capturing the shift mentioned in \cite{Guyon, Guyon2}. This confirms the relevance, in our affine framework, of rough volatility and clustering jumps to model simultaneously the S\&P 500 and VIX dynamics.

The affine relation between variance swap rates and forward variance -- which generalizes the affine relation between variance swap rates and spot variance in the classical framework \cite{Kallsen} -- is a by-product of our affine Volterra framework. This affine relation has been confirmed empirically in \cite{MST20}.

To conclude, the rough Hawkes Heston model is able -- in a tractable and parsimonious fashion -- to jointly calibrate S\&P 500 and VIX options. The parsimonious character of our model is an advantage compared to other models that jointly calibrate SPX/VIX options with either a large number of parameters \cite{ContKok13,Guyon3} or based on martingale transport considerations \cite{Guyon2}. The affine character of the rough Hawkes Heston model allows fast pricing using Fourier-techniques, instead of Monte Carlo or machine learning methods as those used for instance in \cite{Gatheral, rosenbaum_zhang}. Moreover, all the parameters in our model have a financial interpretation, and a complete sensitivity analysis shows that they are not redundant since each of them controls a different feature of the S\&P 500 and VIX volatility smiles.

%\section*{Data Availability Statement}
%The data that support the findings of this study are available from the corresponding author upon reasonable request. 

\appendix

\section{Proof of Theorem \ref{gl_log}}\label{ap_A}
In this appendix we prove Theorem \ref{gl_log} regarding the Riccati-Volterra equation \eqref{F_log}-\eqref{RV_log} used to study the Fourier-Laplace transform of the log returns $(X_t)_{t\ge0}$. {\color{black}
We use the following notation:
given $u,\,v\in\mathbb{C},$ let $[u,v]$ be the segment in $\mathbb{C}$ having $u$ and $v$ as endpoints, i.e. $[u,v]=\{z\in\mathbb{C}: z=(1-t)u+tv,\,t\in[0,1]\}$, and denote by $u\vee v=\mathfrak{R}u\vee \mathfrak{R}v+i \,\mathfrak{Im}u\vee\mathfrak{Im}v.$}
\begin{proof}
	Fix $w\in\mathbb{C}$ with  $\mathfrak{R}w\in[0,1]$.

	\emph{\ref{1t1}}  The proof of this point is divided into three steps. In the first step, we show the existence of a noncontinuable solution $\psi_w$ of \eqref{RV_log}. In the second step, we prove that $\psi_w$ does not explode in finite time, i.e., that it is global solution. To conclude, in the third and last step, we prove the uniqueness of $\psi_w$.
	
	\underline{\emph{Step \upperRomannumeral{1}}.}  Let us compute from \eqref{F_log}, for every $v\in\mathbb{C}_-$,
	\begin{multline}\label{r_F}
		\mathfrak{R}{F}\left(w,v\right)=\frac{1}{2}\left(\left|\mathfrak{R}w\right|^2-\mathfrak{R}w\right)+\left(b+\rho\sqrt{c}\,\mathfrak{R}w\right)\mathfrak{R}v
		+\frac{c}{2}\left|\mathfrak{R}v\right|^2-\frac{1}{2}\left(\left|\mathfrak{Im}w\right|^2+c\left|\mathfrak{Im}v\right|^2+2\rho\sqrt{c}\,\mathfrak{Im}w\mathfrak{Im}v\right)\\
		+\int_{\mathbb{R}_+}\left[e^{\left(\mathfrak{R}v-\Lambda \mathfrak{R}w\right)z}\cos\left(\left(\mathfrak{Im}v-\Lambda \mathfrak{Im}w\right)z\right)-\mathfrak{R}w\left(e^{-\Lambda z}-1\right)-1-\mathfrak{R}vz\right]\nu\left(\dd z\right).
	\end{multline}
	Since  $|\rho|\le 1$ we have $|\rho\sqrt{c}\,\mathfrak{Im}w\mathfrak{Im}v|\le\sqrt{c}|\mathfrak{Im}w||\mathfrak{Im}v|$, which implies 
	\begin{equation}\label{contr}
		-\frac{1}{2}\left(\left|\mathfrak{Im}w\right|^2+c\left|\mathfrak{Im}v\right|^2+2\rho\sqrt{c}\,\mathfrak{Im}w\mathfrak{Im}v\right)
		\le 
		-\frac{1}{2}\left(\left|\mathfrak{Im}w\right|-\sqrt{c}\left|\mathfrak{Im}v\right|\right)^2\le 0.
	\end{equation}
	Recalling that $\mathfrak{R}w\in[0,1]$, we then obtain
	\begin{align}\label{st1}
		\notag&
		\mathfrak{R}{F}\left(w,v\right)\le
		\left(b+\rho\sqrt{c}\,\mathfrak{R}w\right)\mathfrak{R}v+\frac{c}{2}\left|\mathfrak{R}v\right|^2+\int_{\mathbb{R}_+}\left[e^{-\Lambda\mathfrak{R}wz}-\mathfrak{R}w\left(e^{-\Lambda z}-1\right)-1\right]\nu\left(\dd z\right)\notag\\&\qquad\qquad\qquad\qquad\qquad\qquad\qquad\qquad\qquad\qquad\qquad\qquad\notag
		+ \int_{\mathbb{R}_+}\left[e^{\left(\mathfrak{R}v-\Lambda\mathfrak{R}w\right)z}-e^{-\Lambda\mathfrak{R}wz}-\mathfrak{R}vz\right]\nu\left(\dd z\right)\\&
		\le 
		\left(b+\rho\sqrt{c}\,\mathfrak{R}w+\int_{\mathbb{R}_+}z\left(e^{-\Lambda\mathfrak{R}wz}-1\right)\nu\left(\dd z\right)\right)\mathfrak{R}v+\frac{c}{2}\left|\mathfrak{R}v\right|^2
		+ \int_{\mathbb{R}_+}e^{-\Lambda\mathfrak{R}wz}\left(e^{\mathfrak{R}vz}-1-\mathfrak{R}vz\right)\nu\left(\dd z\right),
	\end{align}
	where for the second inequality we use 
	\begin{equation}\label{bah}
		e^{-\Lambda\mathfrak{R}wz}-\mathfrak{R}w\left(e^{-\Lambda z}-1\right)-1\le0,\quad z\ge 0.
	\end{equation}
	Let $h\colon\mathbb{R}_+\times \mathbb{R}_-\to\mathbb{R}_-$ be the continuous function defined by 
	\begin{equation*}
		h\left(x,y\right)=
		\begin{cases}
			\frac{1}{y}\int_{\mathbb{R}_+}e^{-\Lambda x z}\left(e^{yz}-1-yz\right)\nu\left(\dd z\right),&y<0\\
			0,&y=0
		\end{cases},\quad x\ge0,
	\end{equation*} 
	and note that $y\cdot h(x,y)=\int_{\mathbb{R}_+}e^{-\Lambda xz}(e^{yz}-1-yz)\nu(\dd z).$ At this point, we can use \eqref{st1} to show that
	\begin{equation}\label{b1}
		\mathfrak{R}{F}\left(w,v\right)\le 
		\left(C_w+\frac{c}{2}\mathfrak{R}v+h\left(\mathfrak{R}w,\mathfrak{R}v\right)\right)\mathfrak{R}v,\quad v\in\mathbb{C}_-,
	\end{equation}
	where $C_w=b+\rho\sqrt{c}\,\mathfrak{R}w+\int_{\mathbb{R}_+}z(e^{-\Lambda \mathfrak{R}wz}-1)\nu(\dd z)$.
	
	We now introduce the function $\widetilde{F}_w\colon\mathbb{C}\to\mathbb{C}$ given by 
	$$\widetilde{F}_w\left(v\right)=F\left(w,-\mathfrak{R}v^-+i\mathfrak{Im}v\right)+C_w\mathfrak{R}v^+,\quad v\in\mathbb{C}.$$
	Observe that, by construction (see also \eqref{b1})
	\[
	\mathfrak{R}\widetilde{F}_w\left(v\right)
	\le 
	\left(C_w-\frac{c}{2}\mathfrak{R}v^-+h\left(\mathfrak{R}w,-\mathfrak{R}v^-\right)\right)\mathfrak{R}v
	,\quad v\in\mathbb{C}.
	\] 
	Since $\widetilde{F}_w$ is continuous, we can invoke \cite[Therorem $1.1$, Chapter $12$]{g} to assert the existence of a continuous, noncontinuable solution $\psi_w\colon[0,T_{\text{max}})\to\mathbb{C}$ of the equation 
	\begin{equation}\label{st0}
		\chi=K\ast \widetilde{F}_w\left(\chi\left(\cdot\right)\right),\quad t\in\left[0,T_{\text{max}}\right),
	\end{equation} 
	for some $T_{\text{max}}\in(0,\infty].$ If we can show that $\mathfrak{R}\psi_w\le 0$ in $[0,T_{\text{max}})$, then we conclude that $\psi_w$ is indeed a noncontinuable solution of \eqref{RV_log}, as well.
	To this end, consider the continuous function $\zeta(t)=C_w-\frac{c}{2}\mathfrak{R}\psi_w(t)^-+h(\mathfrak{R}w,-\mathfrak{R}\psi_w(t)^-)$ defined for $t\in[0,T_{\text{max}})$. Taking the real part in \eqref{st0}, for every $T\in(0,T_\text{max})$, we obtain
	\begin{equation*}
		\mathfrak{R}\psi_w\left(t\right)=-\gamma_T\left(t\right)
		+\int_{0}^{t}K\left(t-s\right) 
		\zeta\left(s\right)\mathfrak{R}\psi_w\left(s\right)
		\dd s,\quad t\in\left[0,T\right],
	\end{equation*}
	where $\gamma_T(t)=\int_{0}^{t}K(t-s)1_{\{s\le T\}}(\zeta(s)\mathfrak{R}\psi_w(s)-\mathfrak{R}\widetilde{F}_w(\psi_w(s)))\dd s$. By \cite[Remark B.$6$]{ee} $\gamma_T\in\mathcal{G}_K$ (recall \eqref{invar}), and we can invoke \cite[Theorem C.$1$]{ee} to infer that $\mathfrak{R}\psi_w\le 0$ in $[0,T]$. Given that  $T$ was arbitrary, such an inequality holds in the whole interval $[0,T_\text{max})$, completing the first step of the proof.

	\underline{\emph{Step \upperRomannumeral{2}}.} Our goal here is to show that $T_{\text{max}}=\infty$. Let us fix again a generic $T\in(0,T_\text{max})$. Taking the imaginary part in  \eqref{F_log} and \eqref{RV_log} we have, on the interval $[0,T]$,
	\begin{multline}
		\label{im}
		\mathfrak{Im}\psi_w=
		K\ast
		\Bigg[\left(\mathfrak{R}w-\frac{1}{2}\right)\mathfrak{Im}w+\left(b+\rho\sqrt{c}\,\mathfrak{R}w\right)\,\mathfrak{Im}\psi_w+\rho\sqrt{c}\,\mathfrak{Im}w\,\mathfrak{R}\psi_w+c\,\mathfrak{R}\psi_w\,\mathfrak{Im}\psi_w\\
		+
		\int_{\mathbb{R}_+}\left(e^{\mathfrak{R}\left(\psi_w-\Lambda w\right)\cdot z}\sin\left(\mathfrak{Im}\left(\psi_w-\Lambda w\right)\cdot z\right)-\mathfrak{Im}w\left(e^{-\Lambda z}-1\right)-\mathfrak{Im}\psi_w\cdot z\right)\nu\left(\dd z\right)
		\Bigg].
	\end{multline}
	Consider the function $d\colon \mathbb{R}_-\times \mathbb{R}\to \mathbb{R}$ defined as follows
	\begin{equation*}
		d\left(x,y\right)=\begin{cases}
			\frac{1}{y}\int_{\mathbb{R}_+}e^{ x z}\left(\sin\left(y\, z\right)-y\, z\right)\nu\left(\dd z\right),&y\neq0\\
			0,&y=0
		\end{cases},\quad x\le 0.
	\end{equation*}
	Note that $d$ is continuous and nonpositive in its domain. Moreover, by construction 
	\[
	y\cdot d\left(x,y\right)=\int_{\mathbb{R}_+}e^{ x z}\left(\sin\left(y\, z\right)-y\, z\right)\nu\left(\dd z\right),\quad \left(x,y\right)\in\mathbb{R}_-\times \mathbb{R}.
	\]
	To shorten the notation we define $\widetilde{\psi_w}=\psi_w-\Lambda w$. Using the function $d$ we rewrite \eqref{im} as
	\begin{align*}
		&\mathfrak{Im}\psi_w+\frac{\rho^+}{\sqrt{c}}\mathfrak{Im}w=\frac{\rho^+}{\sqrt{c}}\mathfrak{Im}w+
		K\ast
		\Bigg[\left(\mathfrak{R}w-\frac{1}{2}-\int_{\mathbb{R}_+}\left(e^{-\Lambda z}-1+\Lambda z\right)\nu\left(\dd z\right)
		-\frac{\rho^+}{\sqrt{c}}\left(b+\rho\sqrt{c}\,\mathfrak{R}w\right)\right)\mathfrak{Im}w\\&\qquad
		+\left(-\rho^-\sqrt{c}\,\mathfrak{R}\psi_w
		-\left(\Lambda+\frac{\rho^+}{\sqrt{c}}\right)\int_{\mathbb{R}_+} z\left(e^{\mathfrak{R}\widetilde{\psi_w}\cdot z}-1\right)\nu\left(\dd z\right)
		-\left(\Lambda+\frac{\rho^+}{\sqrt{c}}\right) d\left(\mathfrak{R}\widetilde{\psi_w},\mathfrak{Im}\widetilde{\psi_w}\right)\right)\mathfrak{Im}w
		\\&
		\qquad+\left(\left(b+\rho\sqrt{c}\,\mathfrak{R}w\right)+c\,\mathfrak{R}\psi_w
		+
		\int_{\mathbb{R}_+}z\left(e^{\mathfrak{R}\widetilde{\psi_w}\cdot z}-1\right)\nu\left(\dd z\right)
		+
		d\left(\mathfrak{R}\widetilde{\psi_w},\mathfrak{Im}\widetilde{\psi_w}\right)\right)\left(\mathfrak{Im}\psi_w+\frac{\rho^+}{\sqrt{c}}\mathfrak{Im}w\right)
		\Bigg]\\
		&
		\quad=\frac{\rho^+}{\sqrt{c}}\mathfrak{Im}w
		\\&\qquad+
		K\ast \left[\left(C_1-\frac{\rho^+}{\sqrt{c}}\left(b+\rho\sqrt{c}\,\mathfrak{R}w\right)\right)\!\mathfrak{Im}w+f_1\left(\cdot\right)\mathfrak{Im}w+\left(b+\rho\sqrt{c}\,\mathfrak{R}w+f_2\left(\cdot\right)\right)\left(\mathfrak{Im}\psi_w+\frac{\rho^+}{\sqrt{c}}\mathfrak{Im}w\right)\right]\!,
	\end{align*}
	which holds on $[0,T]$. In particular, note that $f_1\ge0$ and $f_2\le 0$ in $[0,T]$. We want to find a continuous function $u\colon\mathbb{R}_+\to\mathbb{R}_+$ such that $|\mathfrak{Im}\psi_w|\le u$ on $[0,T]$. To do this, we argue by cases on $\mathfrak{Im}w$. In the following, we denote  $\widetilde{\Lambda}=\max\{\rho^- c^{-1/2},\,\Lambda\}$. All the claims regarding the sign of solutions to linear Volterra equations are justified by  \cite[Theorem C.$1$]{ee}.
	
	If $\mathfrak{Im}w\ge0$,  then we can consider the unique, nonnegative, continuous solution $l_1\colon[0,T]\to \mathbb{R}_+$ of the linear equation
	\[
	l_1=\frac{\rho^+}{\sqrt{c}}\mathfrak{Im}w+
	K\ast \left[\left|C_1-\frac{\rho^+}{\sqrt{c}}\left(b+\rho\sqrt{c}\,\mathfrak{R}w\right)\right|\mathfrak{Im}w+\left(\left(b+\rho\sqrt{c}\,\mathfrak{R}w\right)+f_2\right)l_1\right].
	\]
	Since the function $\mathfrak{Im}\psi_w+\frac{\rho^+}{\sqrt{c}}\mathfrak{Im}w+l_1$ satisfies -- in $[0,T]$ -- the linear equation
	\[
	\chi=2\frac{\rho^+}{\sqrt{c}}\mathfrak{Im}w+
	K\ast\left[2\,\left(C_1-\frac{\rho^+}{\sqrt{c}}\left(b+\rho\sqrt{c}\,\mathfrak{R}w\right)\right)^+\mathfrak{Im}w
	+f_1\,\mathfrak{Im}w
	+
	\left(\left(b+\rho\sqrt{c}\,\mathfrak{R}w\right)+f_2\right)\chi
	\right],
	\]
	we deduce that $\mathfrak{Im}\psi_w\ge-l_1-\frac{\rho^+}{\sqrt{c}}\mathfrak{Im}w$ on $[0,T]$. Next, we introduce the unique, nonnegative, continuous solution $\widebar{l_1}\colon\mathbb{R}_+\to\mathbb{R}_+$ of the linear equation
	\begin{equation}\label{lo1}
		\widebar{l_1}=\frac{\rho^+}{\sqrt{c}}\left|\mathfrak{Im}w\right|+K\ast \left[
		\left|C_1-\frac{\rho^+}{\sqrt{c}}\left(b+\rho\sqrt{c}\,\mathfrak{R}w\right)\right|\left|\mathfrak{Im}w\right|+\left(b+\rho\sqrt{c}\,\mathfrak{R}w\right)\widebar{l_1}
		\right]
	\end{equation}
	and observe that $\widebar{l_1}-l_1\ge0$ on $[0,T]$, because $\widebar{l_1}-l_1$ solves on $[0,T]$
	\[
	\chi=K\ast\left[-f_2\,l_1+
	\left(b+\rho\sqrt{c}\,\mathfrak{R}w\right)\chi
	\right].
	\]
	Hence, $\mathfrak{Im}\psi_w\ge-\widebar{l_1}-\frac{\rho^+}{\sqrt{c}}|\mathfrak{Im}w|$ on $[0,T]$.  We now focus on the upper bound. Observe that
	\begin{multline*}
		\mathfrak{Im}\psi_w-\widetilde{\Lambda}\,\mathfrak{Im}w=-\widetilde{\Lambda}\,\mathfrak{Im}w+K\ast\Bigg[\left(C_1
		%\mathfrak{R}w-\frac{1}{2}-\int_{\mathbb{R}_+}\left(e^{-\Lambda z}-1+\Lambda z\right)\nu\left(\dd z\right)
		+\left(b+\rho\sqrt{c}\,\mathfrak{R}w\right)\widetilde{\Lambda}\right)\mathfrak{Im}w
		+\left(b+\rho\sqrt{c}\,\mathfrak{R}w+f_2\right)\left(\mathfrak{Im}\psi_w-\widetilde{\Lambda}\,\mathfrak{Im}w\right)\\
		+\left(\left(\widetilde{\Lambda}c+\rho\sqrt{c}\right)\mathfrak{R}\psi_w
		+\left(\widetilde{\Lambda}-\Lambda\right)\left(\int_{\mathbb{R}_+} z\left(e^{\mathfrak{R}\widetilde{\psi_w}\cdot z}-1\right)\nu\left(\dd z\right)+ d\left(\mathfrak{R}\widetilde{\psi_w},\mathfrak{Im}\widetilde{\psi_w}\right)\right)\right)\mathfrak{Im}w
		\Bigg].
	\end{multline*}
	We then take the unique, nonnegative, continuous solution $u_1\colon[0,T]\to\mathbb{R}_+$ of the linear equation
	\[
	u_1=\widetilde{\Lambda}\,\mathfrak{Im}w+K\ast \left[\left|C_1+\left(b+\rho\sqrt{c}\,\mathfrak{R}w\right)\widetilde{\Lambda}\right|\mathfrak{Im}w+\left(b+\rho\sqrt{c}\,\mathfrak{R}w+f_2\right)u_1\right].
	\]
	We infer that $u_1-(\mathfrak{Im}\psi_w-\widetilde{\Lambda}\,\mathfrak{Im}w)\ge0$ since $\widetilde{\Lambda}c+\rho\sqrt{c},\,\widetilde{\Lambda}-\Lambda\ge0$, and $u_1-(\mathfrak{Im}\psi_w-\widetilde{\Lambda}\,\mathfrak{Im}w)$ satisfies (on $[0,T]$)
	\begin{multline*}
		\chi=2\widetilde{\Lambda}\,\mathfrak{Im}w+K\ast \Bigg[2\left(C_1+\left(b+\rho\sqrt{c}\,\mathfrak{R}w\right)\widetilde{\Lambda}\right)^-\mathfrak{Im}w+\left(b+\rho\sqrt{c}\,\mathfrak{R}w+f_2\right)\chi\\
		-
		\left(\left(\widetilde{\Lambda}c+\rho\sqrt{c}\right)\mathfrak{R}\psi_w
		+\left(\widetilde{\Lambda}-\Lambda\right)\left(\int_{\mathbb{R}_+} z\left(e^{\mathfrak{R}\widetilde{\psi_w}\cdot z}-1\right)\nu\left(\dd z\right)+ d\left(\mathfrak{R}\widetilde{\psi_w},\mathfrak{Im}\widetilde{\psi_w}\right)\right)\right)\mathfrak{Im}w
		\Bigg].
	\end{multline*}
	To end, we introduce the unique, nonnegative, continuous solution $\widebar{u_1}\colon\mathbb{R}_+\to\mathbb{R}_+$ of the linear equation
	\begin{equation}\label{up1}
		\widebar{u_1}=\widetilde{\Lambda}\left|\mathfrak{Im}w\right|+K\ast \left[
		\left|C_1+\left(b+\rho\sqrt{c}\,\mathfrak{R}w\right)\widetilde{\Lambda}\right|\left|\mathfrak{Im}w\right|+\left(b+\rho\sqrt{c}\,\mathfrak{R}w\right)\widebar{u_1}
		\right],
	\end{equation}
	and since $\widebar{u_1}-u_1$ satisfies the linear equation $\chi=K\ast [-f_2\,u_1+(b+\rho\sqrt{c}\,\mathfrak{R}w)\chi]$ on $[0,T]$, we conclude that $\widebar{u_1}\ge u_1$ on the same interval. Therefore, $\mathfrak{Im}\psi_w\le\widebar{u_1}+\widetilde{\Lambda}\,\mathfrak{Im}w$ on $[0,T]$.
	
	In the case $\mathfrak{Im}w\le0$ the argument is analogous, but the upper  and lower bounds are inverted. Specifically, with the same steps as the ones just carried out, we have $-\widebar{u_1}-\widetilde{\Lambda}|\mathfrak{Im}w|\le\mathfrak{Im}\psi_w\le \widebar{l_1}+\frac{\rho^+}{\sqrt{c}}|\mathfrak{Im}w|$ on $[0,T]$.
	
	Therefore, defining the continuous function $u\colon\mathbb{R}_+\to\mathbb{R}_+$ by $u=\widebar{l_1}+\widebar{u_1}+(\widetilde{\Lambda}+\frac{\rho^+}{\sqrt{c}})|\mathfrak{Im}w|$, we have
	\begin{equation}\label{bound_im}
		\left|\mathfrak{Im}\psi_w\left(t\right)\right|\le u\left(t\right),\quad 0\le t\le T.
	\end{equation}
	
	Taking the real part in \eqref{RV_log} and using  \eqref{r_F} we deduce that
	\begin{align*}
		\mathfrak{R}\psi_w
		=
		K\ast 
		&\Bigg[
		\frac{1}{2}\left(\left|\mathfrak{R}w\right|^2-\mathfrak{R}w\right)
		+\left(b+\rho\sqrt{c}\,\mathfrak{R}w\right)\mathfrak{R}\psi_w
		+\frac{c}{2}\left|\mathfrak{R}\psi_w\right|^2
		\\&-\frac{1}{2}\left(\left|\mathfrak{Im}w\right|^2+{c}\,\left|\mathfrak{Im}\psi_w\right|^2	+2\rho\sqrt{c}\,\mathfrak{Im}w\,\mathfrak{Im}\psi_w\right)
		-\left|\int_{\mathbb{R}_+}
		e^{\mathfrak{R}\widetilde{\psi_w}\cdot z}\left(\cos\left(\mathfrak{Im}\widetilde{\psi_w}\cdot z\right)-1\right)\nu\left(\dd z\right)\right|
		\\&+
		\int_{\mathbb{R}_+}\left(e^{\mathfrak{R}{\psi_w}\cdot z}\left(
		e^{-\Lambda\mathfrak{R}w z}-1
		\right)-\mathfrak{R}w\left(e^{-\Lambda z}-1\right)\right)\nu\left(\dd z\right)
		+\int_{\mathbb{R}_+}\left(e^{\mathfrak{R}{\psi_w}\cdot z}-1-\mathfrak{R}\psi_w\cdot z\right)\nu\left(\dd z\right)
		\Bigg]
	\end{align*} 
	on $[0,T]$. Since $|\cos(x)-1|=1-\cos (x)\le{x^2}/{2},\,x\in\mathbb{R}$, by  \eqref{bound_im} we have
	\begin{multline}\label{b_1}
		\left|	\int_{\mathbb{R}_+}e^{\mathfrak{R}\widetilde{\psi_w}\cdot z}\left(\cos\left(\mathfrak{Im}\widetilde{\psi_w}\cdot z\right)-1\right)\nu\left(\dd z\right)\right|
		\le \frac{1}{2}\left(\int_{\mathbb{R}_+} \left|z\right|^2\nu\left(\dd z\right)\right)\left|\mathfrak{Im}\widetilde{\psi_w}\right|^2
		\\
		\le \left(\int_{\mathbb{R}_+} \left|z\right|^2\nu\left(\dd z\right)\right)\left(u^2+\Lambda^2\left|\mathfrak{Im}w\right|^2\right),\quad \text{on }\left[0,T\right].
	\end{multline}
	Moreover, notice that by \eqref{bound_im}, since $|\rho|\le1$
	\begin{equation}\label{eccl}
		\frac{1}{2}\left|\left|\mathfrak{Im}w\right|^2+{c}\,\left|\mathfrak{Im}\psi_w\right|^2	+2\rho\sqrt{c}\,\mathfrak{Im}w\mathfrak{Im}\psi_w\right|\le\frac{1}{2} \left(\left|\mathfrak{Im}w\right|+\sqrt{c}\left|\mathfrak{Im}\psi_w\right|\right)^2
		\le \left|\mathfrak{Im}w\right|^2+cu^2
		.
	\end{equation}
	These facts coupled with \eqref{bah} suggest to consider the linear equation
	\begin{multline}\label{lot}
		l=
		K\ast\Bigg[\frac{1}{2}\left(\left|\mathfrak{R}w\right|^2-\mathfrak{R}w-2\left|\mathfrak{Im}w\right|^2\right)+\int_{\mathbb{R}_+}\left(e^{-\Lambda\mathfrak{R}wz}-1-\mathfrak{R}w\left(e^{-\Lambda z}-1\right)\right)\nu\left(\dd z\right)-c\,u^2
		\\-\left(\int_{\mathbb{R}_+}\left|z\right|^2\nu\left(\dd z\right)\right)\left(u^2+\Lambda^2\left|\mathfrak{Im}w\right|^2\right)
		+\left(b+\rho\sqrt{c}\,\mathfrak{R}w\right)l
		\Bigg],
	\end{multline}
	which has a unique, continuous, nonpositive solution $l$ defined on the whole $\mathbb{R}_+.$ At this point, observe that the difference $\mathfrak{R}\psi_w-l$ satisfies the linear equation
	\begin{align*}
		\chi=K\ast
		\Bigg[&\left(b+\rho\sqrt{c}\,\mathfrak{R}w\right)\chi+\frac{c}{2}\left|\mathfrak{R}\psi_w\right|^2+
		\left(\left|\mathfrak{Im}w\right|^2+cu^2-\frac{1}{2}\left(\left|\mathfrak{Im}w\right|^2+{c}\,\left|\mathfrak{Im}\psi_w\right|^2	+2\rho\sqrt{c}\,\mathfrak{Im}w\mathfrak{Im}\psi_w\right)\right)\\&
		+
		\int_{\mathbb{R}_+}\left(e^{\mathfrak{R}\psi_w\cdot z}-1-\mathfrak{R}\psi_w\cdot z\right)\nu\left(\dd z\right)
		+\int_{\mathbb{R}_+}\left(e^{\mathfrak{R}\psi_w\cdot z}-1\right)\left(e^{-\Lambda\mathfrak{R}wz}-1\right)\nu\left(\dd z\right)
		\\&
		+\left(\left(\int_{\mathbb{R}_+} \left|z\right|^2\nu\left(\dd z\right)\right)\left(u^2+\Lambda^2\left|\mathfrak{Im}w\right|^2\right)-	\left|\int_{\mathbb{R}_+}e^{\mathfrak{R}\widetilde{\psi_w}\cdot z}\left(\cos\left(\mathfrak{Im}\widetilde{\psi_w}\cdot z\right)-1\right)\nu\left(\dd z\right)\right|\right)
		\Bigg].
	\end{align*}
	It admits a unique, continuous solution on $[0,T]$ which is nonnegative by \eqref{b_1}, \eqref{eccl} and the fact that $e^x-1-x\ge0,\,x\in\mathbb{R}.$ Since $T\in(0,T_{\text{max}})$ was chosen arbitrarily, we infer that
	\[
	l\left(t\right)\le \mathfrak{R}\psi_w\left(t\right)\le0\text{\quad and \quad  }\left|\mathfrak{Im}\psi_w\left(t\right)\right|\le u\left(t\right),\quad 0\le t< T_\text{max}.
	\]
	Recalling that $l$ and $u$ are continuous on $\mathbb{R}_+$, and in particular bounded on every compact interval, we conclude that $T_\text{max}=\infty$, as desired.
	
	\underline{\emph{Step \upperRomannumeral{3}}.} Consider two global solutions $\psi_w,\,\psi'_w$ of \eqref{RV_log}, and let $\delta=\psi_w-\psi'_w$ and  $\tilde{\delta}=\psi'_w\vee \psi_w$. Then, for every $t\ge0$,
	\begin{multline}\label{u1}
		\delta\left(t\right)=\int_{0}^{t}K\left(t-s\right)\Bigg[\left(b+\rho\sqrt{c}\,w+\frac{c}{2}\left(\psi_w+\psi'_w\right)\left(s\right)
		+\int_{\mathbb{R}_+}z\left(e^{\left(-\Lambda w+\tilde{\delta}\left(s\right)\right)z}-1\right)\nu\left(\dd z\right)
		\right)\delta\left(s\right)\\+
		\int_{\mathbb{R}_+}e^{\left(-\Lambda w+\tilde{\delta}\left(s\right)\right) z}\left(e^{\left(\psi_w-\tilde{\delta}\right)\left(s\right)z}-e^{\left(\psi'_w-\tilde{\delta}\right)\left(s\right)z}-\delta\left(s\right)z\right)\nu\left(\dd z\right)
		\Bigg]\dd s.
	\end{multline}
	We introduce the function $k_w\colon \mathbb{C}_-\times \mathbb{C}_-\to\mathbb{C}$ defined for $(u,v)\in\mathbb{C}_-\times \mathbb{C}_-$ by 
	\begin{equation}\label{k_def}
		k_w\left(u,v\right)=
		\begin{cases}
			\frac{1}{v-u}\int_{\mathbb{R}_+}e^{\left(-\Lambda w+u\vee v\right)z}\left(e^{\left(v-u\vee v\right) z}-e^{\left(u-u\vee v\right)z}-\left(v-u\right)z\right)\nu\left(\dd z\right),&u\neq v\\
			0,&\text{otherwise}
		\end{cases}.
	\end{equation}
	We claim that $k_w$ is continuous on its domain. This is a consequence of an application of the mean value theorem to the functions $f_z(u)=e^{uz}-uz,\,u\in\mathbb{C}_-$, with the parameter $z\in\mathbb{R}_+$.  Indeed, using the inequality $|1-\cos x|\le x^2,\,x\in\mathbb{R}$,
	\begin{align}\label{import}
		\left|f_z\left(v\right)-f_z\left(u\right)\right|&\le z \sup_{\xi\in \left[u,v\right]}\left|e^{\xi z}-1\right|\left|v-u\right|\notag\\\notag&
		\le
		z\sup_{\xi\in \left[u,v\right]}\left(\left|e^{\mathfrak{R}\xi\cdot  z}-1\right|+\sqrt{2}e^{\frac{1}{2}\mathfrak{R}\xi\cdot z}\left(1-\cos\left(\mathfrak{Im}\xi \cdot z\right)\right)^\frac{1}{2}\right)\left|v-u\right|
		\\&
		\le z\left(\left(1-e^{\left(\mathfrak{R}u\wedge\mathfrak{R}v\right)z}\right)+\sqrt{2} \left(\left|\mathfrak{Im}u\right|\vee\left|\mathfrak{Im}v\right|\right)\left|z\right|\right)\left|v-u\right|
		,\quad u,\,v\in \mathbb{C}_-,\,z\in\mathbb{R}_+.
	\end{align}
	Consequently, the continuity of $k_w$ follows from
	\begin{equation}\label{pure}
		\left|f_z\left(v-u\vee v\right)-f_z\left(u-u\vee v\right)\right|\le \left|z\right|^2\left(1+\sqrt{2}\right)\left|v-u\right|^2, \quad u,\,v\in \mathbb{C}_-,\,z\in\mathbb{R}_+.
	\end{equation}
	Coming back to \eqref{u1} we have (on $\mathbb{R}_+$)
	\begin{equation}\label{refer}
		\delta=K\ast\left[\left(b+\rho\sqrt{c}\,w+\frac{c}{2}\left(\psi_w+\psi'_w\right)\left(\cdot\right)
		+\int_{\mathbb{R}_+}z\left(e^{\left(-\Lambda w+\tilde{\delta}\left(\cdot\right)\right)z}-1\right)\nu\left(\dd z\right)\\+k_w\left(\psi'_w\left(\cdot\right),\psi_w\left(\cdot\right)\right)
		\right)\delta\right],
	\end{equation}
	which is a linear equation admitting the zero function as its unique solution. Hence $\psi'_w=\psi_w$ on $\mathbb{R}_+$, completing the proof of this step.
	
	The fact that $\psi_{\mathfrak{R}w}$ is $\mathbb{R}_--$valued follows from \eqref{bound_im}, because in this case $u\equiv 0$. This concludes the proof of the statement in \emph{\ref{1t1}}.
	
	\emph{\ref{2t1}} 
	From \eqref{r_F} and \eqref{contr} we deduce that $\mathfrak{R}{F}(w,v)\le F(\mathfrak{R}w,\,\mathfrak{R}v)$ for every $v\in\mathbb{C}_-$. Taking the real part in \eqref{RV_log} and recalling that -- under Hypothesis \ref{c1} -- the kernel $K$ is nonnegative on $(0,\infty)$ we obtain
	\begin{equation*}
		\mathfrak{R}\psi_w\left(t\right)\le\int_0^t K\left(t-s\right){F}\left(\mathfrak{R}w,\,\mathfrak{R}\psi_w\left(s\right)\right)\dd s,\quad t\ge0.
	\end{equation*}
	We can then introduce a nonnegative function $\widetilde{\gamma}\colon\mathbb{R}_+\to\mathbb{R}_+$ defined by the  relation
	\begin{equation}\label{j1}
		\mathfrak{R}\psi_w\left(t\right)=-\widetilde{\gamma}\left(t\right)+\int_0^t K\left(t-s\right){F}\left(\mathfrak{R}w,\,\mathfrak{R}\psi_w\left(s\right)\right)\dd s,\quad t\ge0.
	\end{equation}
	Using \eqref{RV_log}, one can rewrite $\widetilde{\gamma}$ as  
	\begin{equation*}
		\widetilde{\gamma}\left(t\right)=\int_{0}^{t}K\left(t-s\right)\left({F}\left(\mathfrak{R}w,\mathfrak{R}\psi_w\left(s\right)\right)-\mathfrak{R}F\left(w,\psi_w\left(s\right)\right)\right)	\dd s,\quad t\ge0.
	\end{equation*}
	Thus $\widetilde{\gamma}\in\mathcal{G}_K$ by \cite[Remark B.$6$]{ee}. At this point we subtract \eqref{j1} from \eqref{RV_log} (with $\mathfrak{R}w$ instead of $w$) to deduce that $\delta={\psi_{\mathfrak{R}w}}-\mathfrak{R}\psi_w$ satisfies
	\begin{equation}\label{delta}
		\delta\left(t\right)=\widetilde{\gamma}\left(t\right)+\int_{0}^t K\left(t-s\right)\left({F}\left(\mathfrak{R}w,{\psi_{\mathfrak{R}w}}\left(s\right)\right)
		-{F}\left(\mathfrak{R}w,\,\mathfrak{R}\psi_w\left(s\right)\right)
		\right)\dd s,\quad t\ge0.
	\end{equation}
	If we denote by  $\tilde{\delta}=\mathfrak{R}\psi_w\vee \psi_{\mathfrak{R}w}$, we then need to study (on $\mathbb{R}_+)$
	\begin{align*}
		{F}\left(\mathfrak{R}w,{\psi_{\mathfrak{R}w}}\right)
		-{F}\left(\mathfrak{R}w,\,\mathfrak{R}\psi_w\right)
		&=\left(b+\rho\sqrt{c}\,\mathfrak{R}w+\frac{c}{2}\left(\mathfrak{R}\psi_w+\psi_{\mathfrak{R}w}\right)+\int_{\mathbb{R}_+}z\left(e^{\left(-\Lambda\mathfrak{R}w+\tilde{\delta} \right)z}-1\right)\nu\left(\dd z\right)\right)\delta	\\&\qquad\quad	
		+\int_{\mathbb{R}_+}e^{\left(-\Lambda\mathfrak{R}w+\tilde{\delta}\right)z}\left(e^{\left(\psi_{\mathfrak{R}w}-\tilde{\delta}\right) z}
		-	e^{\left(\mathfrak{R}\psi_w-\tilde{\delta}\right) z}
		-\delta z\right)\nu\left(\dd z\right)
		\\
		&=\left(w_1\left(\cdot\right)+k_{\mathfrak{R}w}\left(\mathfrak{R}\psi_w\left(\cdot\right), \psi_{\mathfrak{R}w}\left(\cdot\right)\right)\right)\delta,
	\end{align*}
	with $k_{\mathfrak{R}w}$ as in \eqref{k_def}.
	Going back to \eqref{delta}, 
	\begin{equation*}
		\delta\left(t\right)=\widetilde{\gamma}\left(t\right)+\int_{0}^tK\left(t-s\right)\left(w_1\left(s\right)+k_{\mathfrak{R}w}\left(\mathfrak{R}\psi_w\left(s\right),\psi_{\mathfrak{R}w}\left(s\right)\right)\right)\delta\left(s\right)\dd s,\quad t\ge0.
	\end{equation*}
	We can now apply \cite[Theorem C.$1$]{ee} in order to conclude that $\delta\ge0$ on $\mathbb{R}_+$. This yields \eqref{comp_log} and concludes the proof of \emph{\ref{2t1}}.
\end{proof}
\section{Proof of Proposition \ref{price_Lewis}}\label{Lew}
This section is devoted to the proof of Proposition \ref{price_Lewis}, a result which allows to price options on the underlying asset $S$ with maturity $T>0$.
\begin{proof}
	Let us define the function $f\colon\mathbb{R}\to \mathbb{R}$ by
	\begin{equation}\label{f}
		f\left(m\right)=\mathbb{E}\left[e^{X_T}-\left(e^{X_T}-e^m\right)^+\right]e^{-\frac{1}{2}m}=\mathbb{E}\left[e^{X_T}1_{\left\{X_T\le m\right\}}+e^m1_{\left\{m<X_T\right\}}\right]e^{-\frac{1}{2}m},\quad m\in\mathbb{R}.
	\end{equation}
	Denote by $\mu_T$ the probability distribution of $X_T$ on $\mathbb{R}$ and note that $f\in L^1(\mathbb{R})$, because, thanks to Tonelli's theorem, 
	\begin{equation}\label{inverti}
		\int_{\mathbb{R}}e^{-\frac{1}{2}m}\left[\int_{\mathbb{R}}\left(e^x1_{\left\{x\le m\right\}}+e^m1_{\left\{m<x\right\}}\right)\mu_T\left(\dd x\right)\right]\dd m
		=4\int_{\mathbb{R}}e^{\frac{1}{2}x}\mu_T\left(\dd x\right)
		=4 \mathbb{E}\left[e^{\frac{1}{2}X_T}\right]<\infty.
	\end{equation}
	Therefore we can compute the Fourier transform of $f$ as follows
	\begin{multline*}
		\hat{f}\left(\lambda\right)=\int_{\mathbb{R}}e^{\left(-\frac{1}{2}+i\lambda\right) m}\left[\int_{\mathbb{R}}\left(e^x1_{\left\{x\le m\right\}}+e^m1_{\left\{m<x\right\}}\right)\mu_T\left(\dd x\right)\right]\dd m
		\\=\int_{\mathbb{R}}\left[e^x\int_{x}^{\infty}e^{\left(-\frac{1}{2}+i\lambda\right) m}\dd m+\int_{-\infty}^xe^{\left(\frac{1}{2}+i\lambda\right) m}\dd m\right]\mu_T\left(\dd x\right)
		=\frac{1}{\frac{1}{4}+\lambda^2}\Psi^{X_T}\left(\frac{1}{2}+i\lambda\right),\quad \lambda\in\mathbb{R},
	\end{multline*}
	where in the second equality we are allowed to use Fubini's theorem by \eqref{inverti}.
	
	Since $|\Psi ^{X_T}(\frac{1}{2}+i\lambda)|\le \mathbb{E}[e^{\frac{1}{2}X_T}]<\infty$ and, by dominated convergence, $f$ in continuous on $\mathbb{R}$, we invoke the Fourier inversion theorem, see for instance \cite[Theorem $9.11$]{rudin}, to obtain
	\begin{equation}\label{inversion}
		f\left(m\right)=\frac{1}{2\pi}\int_{\mathbb{R}}e^{-im\lambda}\frac1{\frac{1}{4}+\lambda^2}\Psi^{X_T}\left(\frac{1}{2}+i\lambda\right)\dd \lambda,\quad m\in\mathbb{R}.
	\end{equation}
	Combining \eqref{f} and \eqref{inversion} and recalling Corollary \ref{cor_mar} we deduce that
	\begin{equation}\label{pre}
		\mathbb{E}\left[\left(e^{X_T}-e^m\right)^+\right]=1-\frac{1}{2\pi}\int_{\mathbb{R}}e^{\left(\frac{1}{2}-i\lambda\right)m}\frac1{\frac{1}{4}+\lambda^2}\Psi^{X_T}\left(\frac{1}{2}+i\lambda\right)\dd \lambda,\quad m\in\mathbb{R}.
	\end{equation}
	Now, for every $k\in \mathbb{R}$, we can determine the  price $C_S(k,T)$ of a call option written on $S$ with log strike $k$ and maturity $T$. Indeed,  
	taking $m=k-\log(S_0)$ in \eqref{pre} we have
	\begin{align*}
		C_S\left(k,T\right)&=\mathbb{E}\left[\left(S_T-e^k\right)^+\right]=S_0
		-\frac{1}{2\pi}\sqrt{S_0e^k}\int_{\mathbb{R}}e^{i\lambda\left(\log\left(S_0\right)-k\right)}\frac1{\frac{1}{4}+\lambda^2}\Psi^{X_T}\left(\frac{1}{2}+i\lambda\right)\dd \lambda\notag\\
		&=S_0-\frac{1}{\pi}\sqrt{S_0e^k}\int_{\mathbb{R}_+}\mathfrak{R}\left[e^{i\lambda\left(\log\left(S_0\right)-k\right)}\Psi^{X_T}\left(\frac{1}{2}+i\lambda\right)\right]\frac1{\frac{1}{4}+\lambda^2}\,\dd \lambda,
	\end{align*}
	which coincides with \eqref{pr_SPX}. The expression \eqref{put_SPX} for the price $P_S(k,T)$ of a put option with the same underlying, log strike and maturity as before, follows from \eqref{pr_SPX}, Corollary \ref{cor_mar}, and the put-call parity formula. This completes the proof.
\end{proof}

\section{Proof of Theorem \ref{stime}}\label{ap_C}
This section is devoted to the proof of Theorem \ref{stime}, a result providing estimates for the multi-factor approximation of the Riccati-Volterra equations appearing in the Fourier-Laplace transform of the log returns and VIX$^2$.
\begin{proof}
		Fix $T>0$. We first prove Point \emph{\ref{t11}}. Take $w\in\mathbb{C}$ such that $\mathfrak{R}w\in[0,1]$ and $n\in\mathbb{N},$ and observe that $|\psi_{w,n}|\le \widebar{l_{1,n}}+\widebar{u_{1,n}}-l_n+(\widetilde{\Lambda}+\frac{\rho^+}{\sqrt{c}})|\mathfrak{Im}w|$ on $\mathbb{R}_+$. Here $\widetilde{\Lambda}=\max\{\rho^-c^{-1/2},\,\Lambda\}$ and $\widebar{l_{1,n}}$ [resp., $\widebar{u_{1,n}},\,l_n$] is the unique, continuous solution of \eqref{lo1} [resp., \eqref{up1}, \eqref{lot}] in Appendix \ref{ap_A} with $K_n$ instead of $K$.
		\cite[Corollary C.$4$]{ee}  guarantees the existence of a positive constant $C_1=C_1(\rho, b,c,\Lambda,\nu)$ such that
		\[
		\widebar{l_{1,n}}\left(t\right)+\widebar{u_{1,n}}\left(t\right)+\left(\widetilde{\Lambda}+\frac{\rho^+}{\sqrt{c}}\right)\left|\mathfrak{Im}w\right|\le C_1\left(1+\int_{0}^{T}\left|E_{b+\rho^+\sqrt{c},n}\left(s\right)\right|\dd s\right)\left|\mathfrak{Im}w\right|,\quad t\in\left[0,T\right].
		\]
		%Since $K_n\to K$ in $L^2\left(\left[0,T\right]\right)$, \cite[Theorem~$3.1$, Chapter $2$]{g}
		Then, recalling the hypothesis of boundedness for $(\int_{0}^{T}|E_{b+\rho^+\sqrt{c},n}(s)|\dd s)_n$ and using \eqref{lot}, another application of \cite[Corollary C.$4$]{ee} provides the existence of a constant $C_2=C_2(\rho,b,c,\Lambda,\nu,\mathbf{m}, \mathbf{x},T)>0$ such that $|l_n(t)|\le C_2(1+\left|\mathfrak{Im}w\right|^2),\,t\in[0,T]$. This implies, given that $n\in\mathbb{N}$ is arbitrary, that
		\begin{equation}\label{le1}
			\sup_{n\in\mathbb{N}}\sup_{t\in\left[0,T\right]}\left|\psi_{w,n}\left(t\right)\right|\le C_3\left(1+\left|\mathfrak{Im}w\right|^2\right),\quad  \text{for some }C_3=C_3\left(\rho,b,c,\Lambda,\nu,\mathbf{m},\mathbf{x},T\right)>0.
		\end{equation}
		Since the same argument works for $\psi_w$, without loss of generality, we assume that the upper bound in \eqref{le1} holds also for $\psi_w$. Now, from \eqref{RV_log} and \eqref{RV_log_n} we have (on $\mathbb{R}_+$)
		\begin{equation*}
			\psi_w-\psi_{w,n}= \left(K-K_n\right)\ast F\left(w,\psi_{w,n}\left(\cdot\right)\right)	+ 
			K\ast \left(F\left(w,\psi_w\left(\cdot\right)\right)-F\left(w,\psi_{w,n}\left(\cdot\right)\right)\right),\quad n\in\mathbb{N}.
		\end{equation*}
		For every $v\in\mathbb{C}_-$, recalling the inequality $e^x-1-x\le x^2/2,\,x\le 0$, and thanks to the computations in Appendix \ref{ap_A} (see \eqref{import})
		\begin{multline*}
			\left|\int_{\mathbb{R}_+}\left[e^{\left(v-\Lambda w\right)z}-w\left(e^{-\Lambda z}-1\right)-1-vz\right]\nu\left(\dd z\right)\right|\\\le 
			4\sqrt{2}\left[\frac{\Lambda^2}{2}\left(1+\left|\mathfrak{Im}w\right|\right)+\left|v\right|^2+\Lambda^2\left(1+\left|\mathfrak{Im}w\right|^2\right)\right]\int_{\mathbb{R}_+}\left|z\right|^2\nu\left(\dd z\right).	
		\end{multline*}
		Then by \eqref{le1} and  \eqref{F_log} we deduce that there exists a constant $C_4=C_4(\rho,b,c,\Lambda, \nu,\mathbf{m},\mathbf{x},T)>0$ such that
		\begin{equation}\label{ste1}
			\sup_{t\in\left[0,T\right]} \left|\left(\left(K-K_n\right)\ast \left(F\left(w,\psi_{w,n}\left(\cdot\right)\right)\right)\right)\left(t\right)\right|\le C_4\left(1+\left|\mathfrak{Im}w\right|^4\right)\int_{0}^{T}\left|K_n\left(s\right)-K\left(s\right)\right|\,\dd s,\quad n\in\mathbb{N}.
		\end{equation}
		In what follows, we denote by $h_n=(K-K_n)\ast F(w,\psi_{w,n}(\cdot)),$ i.e., the function that we have just bounded. Next, computations analogous to those carried out to obtain the Volterra equation \eqref{refer} in Appendix \ref{ap_A},
		allow us to write (on $\mathbb{R}_+$)
		\begin{multline*}
			F\left(w,\psi_w\right)-F\left(w,\psi_{w,n}\right)=\Bigg(b+\rho\sqrt{c}\,w+\frac{c}{2}\left(\psi_w+\psi_{w,n}\right)
			+\int_{\mathbb{R}_+}z\left(e^{\left(-\Lambda w+\psi_{w,n}\vee\psi_{w}\right)z}-1\right)\nu\left(\dd z\right)
			\\+k_w\left(\psi_{w,n},\psi_w\right)\Bigg)\left(\psi_w-\psi_{w,n}\right),
		\end{multline*}
		where $k_w$ is the continuous function  in \eqref{k_def}. Therefore, since $|k_w(u,v)|\le(1+\sqrt{2}) (\int_{\mathbb{R}_+}|z|^2\nu(\dd z))|v-u|$ for every $u,v\in\mathbb{C}_-$ (see \eqref{pure}) and recalling \eqref{le1}-\eqref{ste1}, an application of  \cite[Corollary C.$4$]{ee} yields
		\begin{multline}\label{eh_oh}
			\sup_{t\in\left[0,T\right]}\left|\psi_w\left(t\right)-\psi_{w,n}\left(t\right)-h_n\left(t\right)\right|\le C_5\left(1+\left|\mathfrak{Im}w\right|^6\right) 
			\frac{\int_{0}^{T}E_{b^++\rho^+\sqrt{c}+c_{\nu}C_3\left(1+\left|\mathfrak{Im}w\right|^2\right),K}\left(s\right)\dd s}{\int_{0}^{T}\left|E_{b^++\rho^+\sqrt{c},K}\left(s\right)\right|\dd s}\\
			\times\int_{0}^{T}\left|K_n\left(s\right)-K\left(s\right)\right|\dd s
			,\quad n\in\mathbb{N}.
		\end{multline}
		for some $C_5=C_5(\rho, b,c,\Lambda,\nu,\mathbf{m}, \mathbf{x}, T)>0$ and where $c_\nu=2(1+\sqrt{2})(\int_{\mathbb{R}_{+}}|z|^2\nu(\dd z))$. Notice that by \cite[Proposition $8.1$, Chapter $9$]{g} and Hypothesis \ref{c1}, $E_{b^++\rho^+\sqrt{c}+c_{\nu}C_3(1+|\mathfrak{Im}w|^2),K}\ge0$. Consequently, thanks to \cite[Theorem C.$1$, Remark B.$6$]{ee}, $E_{b^++\rho^+\sqrt{c},K}\le E_{b^++\rho^+\sqrt{c}+c_{\nu}C_3(1+|\mathfrak{Im}w|^2),K}$ a.e. in $\mathbb{R}_+$. Hence the ratio in \eqref{eh_oh} is greater or equal to $1$. Combining \eqref{eh_oh} with \eqref{ste1} yields \eqref{bou_log}. 
		
		In order to prove the final remark about the independence of the constant $C$ in \eqref{bou_log} with respect to $\mathbf{m}$ and $\mathbf{x}$, note that in the previous argument such a dependence is only due to $\widetilde{C},$ the positive constant given by the hypothesis controlling the sequence $(\int_{0}^{T}|E_{b+\rho^+\sqrt{c},n}(s)|\dd s)_n$. When $b<0$, the kernels $-bK_n$  inherit the property of complete monotonicity from $K_n$. If in addition $\rho<0$, we can use \cite[Theorem $3.1$, Chapter~$5$]{g} to infer that $\int_{0}^{T}|E_{b+\rho^+\sqrt{c},n}(s)|\dd s=\int_{0}^{T}|E_{b,n}(s)|\dd s\le |b|^{-1}$ for every $n\in\mathbb{N}$, and $\int_{0}^{T}|E_{b^++\rho^+\sqrt{c},K}(s)|\dd s=\|K\|_{L^1([0,T])}$. In particular, in this case $C$ depends on $T$ only via the $L^1-$norm of $K$ in $[0,T]$ (see \eqref{le1}-\eqref{eh_oh}).
		
		The proof of Point \emph{\ref{t12}} follows by an analogous argument. In this case we use the estimates in \cite[Appendix B.$1$]{primo} and the fact that $\int_{0}^{{\delta}}K_n(s)\dd s\le\int_{0}^{T\vee{\delta}}E_{b^+,n}(s)\dd s\le \widetilde{C}$, $n\in\mathbb{N}$. We also combine \cite[Corollary C.$4$]{ee}, the comparison result for linear Volterra equations in \cite[Theorem~$2$]{comp}, and the inequality 
		\[
		\int_{0}^{{\delta}}h\left(s\right)K_n\left(s+t\right)\dd s
		\le \int_{0}^{{\delta}}h\left(s\right)K_n\left(s\right)\dd s,\quad t\ge0,
		\] 
		which holds also for $K$ by Hypothesis \ref{c1}. 	
		\end{proof}

%%%%% FIGURES %%%%%%%
\newpage 
\hspace{0pt}
\vfill
\begin{center}
\begin{figure}[h]
  %  \centering
   \vspace{-0.4cm} 
     \includegraphics[width=8.3cm]{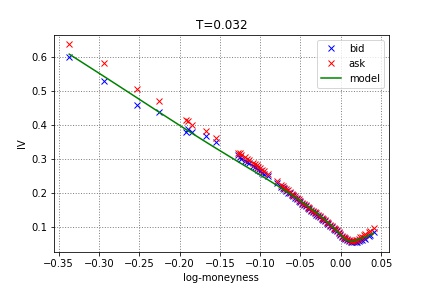} \hspace{-0.75cm}
     \includegraphics[width=8.3cm]{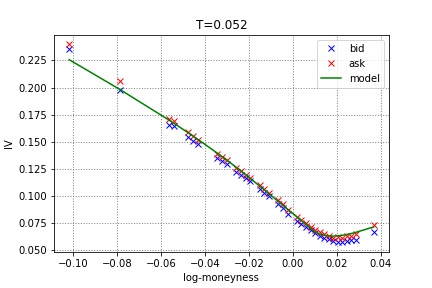}
     \includegraphics[width=8.3cm]{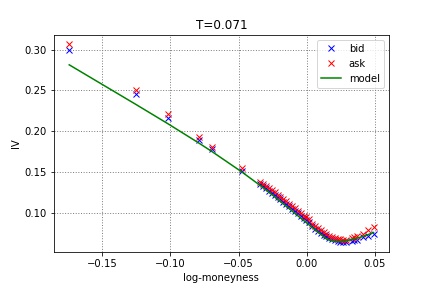} \hspace{-0.75cm}
      \includegraphics[width=8.3cm]{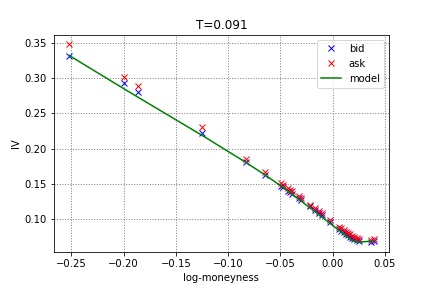} 
    \caption{Calibrated implied volatility of SPX options on 19 May 2017, see Table \ref{Table:CalibratedParameters}. 
    The blue and red crosses are
respectively the bid and ask of market implied volatilities. The implied volatility smiles from the
model are in green. The abscissa is in log-moneyness and T is time to expiry in years.}
  \label{graph-calibration-SP} 
\end{figure}
\end{center}
\hspace{0pt}
\vfill

\newpage
\hspace{0pt}
\vfill
 \begin{center}
\begin{figure}[h]
  %  \centering
   \vspace{-0.4cm} 
     \includegraphics[width=8.3cm]{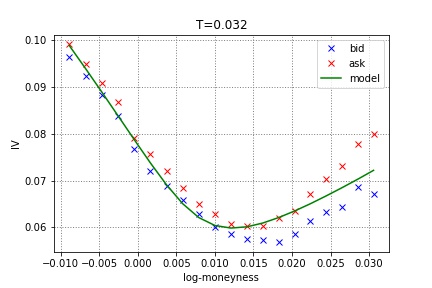} \hspace{-0.75cm}
     \includegraphics[width=8.3cm]{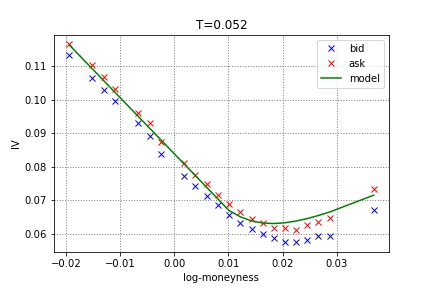}
     \includegraphics[width=8.3cm]{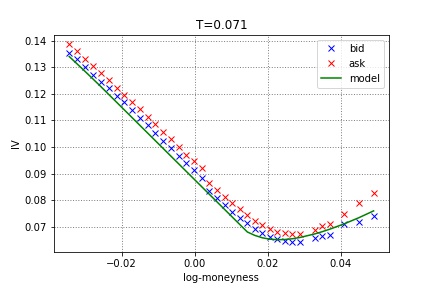} \hspace{-0.75cm}
      \includegraphics[width=8.3cm]{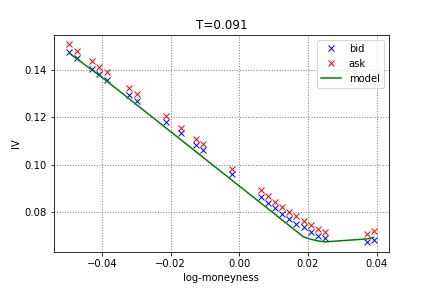} 
    \caption{Zoom around the money of calibrated implied volatility of SPX options on 19 May 2017.}
  \label{graph-calibration-SP-zoom} 
\end{figure}
\end{center}
\hspace{0pt}
\vfill
\newpage
\hspace{0pt}
\vfill
\begin{center}
\begin{figure}[h]
  %  \centering
   \vspace{-0.4cm} 
     \includegraphics[width=8.3cm]{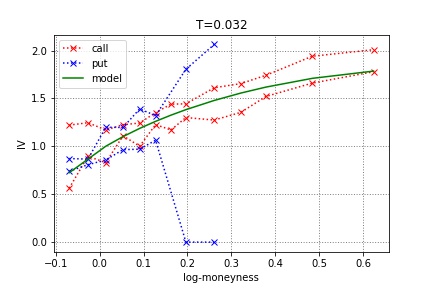} \hspace{-0.75cm}
     \includegraphics[width=8.3cm]{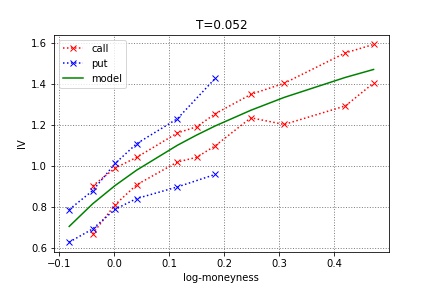}
     \includegraphics[width=8.3cm]{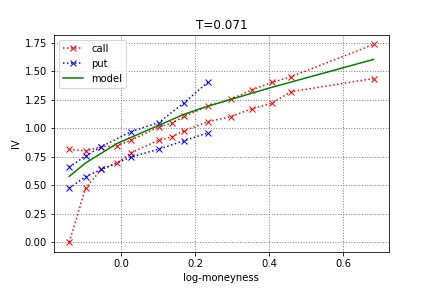} \hspace{-0.75cm}
      \includegraphics[width=8.3cm]{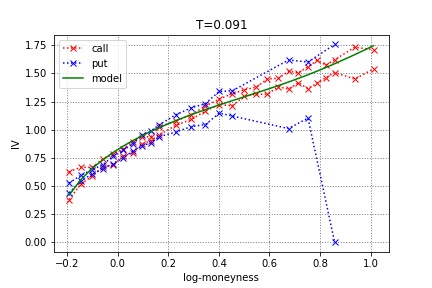} 
    \caption{Calibrated implied volatility of VIX options on 19 May 2017, see Table \ref{Table:CalibratedParameters}. 
    The blue and red crosses are the bid-ask corridors of market implied volatilities computed from put and call options, respectively. The implied volatility smiles from the
model are in green. The abscissa is in log-moneyness and T is time to expiry in years.}
  \label{graph-calibration-VIX}
\end{figure}
\end{center}
\hspace{0pt}
\vfill
\newpage
\hspace{0pt}
\vfill
   \begin{center}
\begin{figure}[h]
  %  \centering
   \vspace{-0.4cm} 
     \includegraphics[width=5cm]{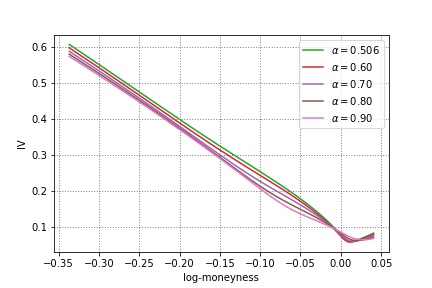} \hspace{-0.3cm}
     \includegraphics[width=5cm]{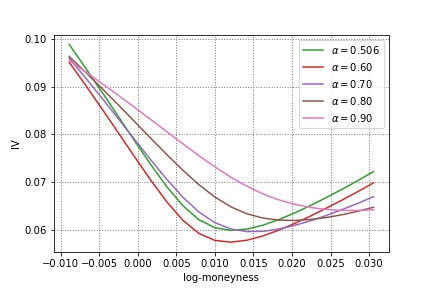} \hspace{-0.3cm}
      \includegraphics[width=5.7cm]{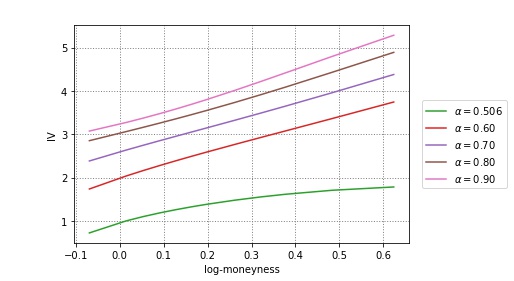} 
    \caption{Sensitivity of implied volatility for SPX (left and center) and VIX (right) options with respect to the kernel power $\alpha$ for the shortest maturity.}
  \label{fig:sensitivity-alpha}
\end{figure}
\end{center}

\begin{center}
\begin{figure}[h]
  %  \centering
   \vspace{-0.4cm} 
     \includegraphics[width=5.4cm]{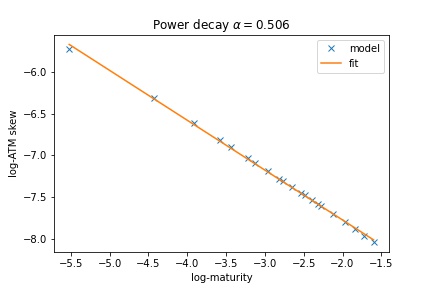} \hspace{-0.75cm}
     \includegraphics[width=5.8cm]{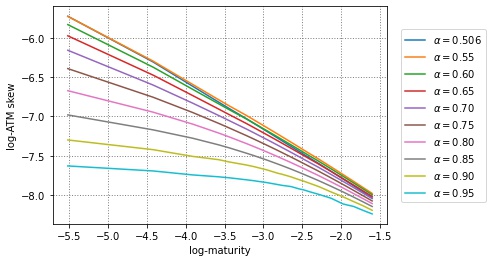} \hspace{-0.25cm}
      \includegraphics[width=5.4cm]{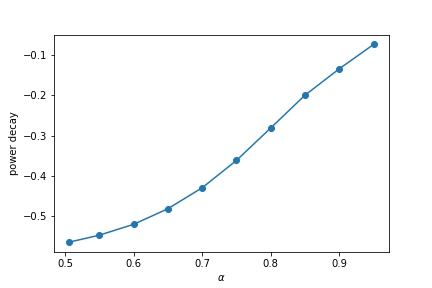} \hspace{-0.8cm}  
    \caption{Power decay of the ATM volatility skew.
     On the left, the log-log plot of ATM volatility skew for the calibrated parameters of Table \ref{Table:CalibratedParameters}. At the center, the log-log plot of ATM volatility skew for different values of $\alpha$; the other parameters are as in Table \ref{Table:CalibratedParameters}. On the right, the fitted power decay of the ATM volatility skew as function of $\alpha$; the power decay is estimated using the five shortest maturities, i.e. $\log(T) \in [-5.5, -3.5]$.}
    \label{graph-skew-decay}
 %       \vspace{-0.6cm} 
\end{figure}
\end{center}
\hspace{0pt}
\vfill
\newpage
\hspace{0pt}
\vfill
 \begin{center}
\begin{figure}[h!]
  %  \centering
   \vspace{-2cm} 
    \includegraphics[width=5cm]{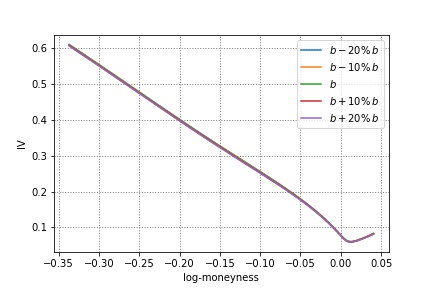} \hspace{-0.3cm}
     \includegraphics[width=5cm]{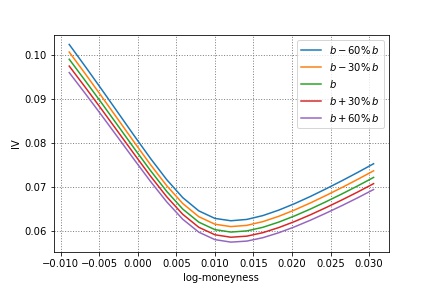} \hspace{-0.3cm}
      \includegraphics[width=5.5cm]{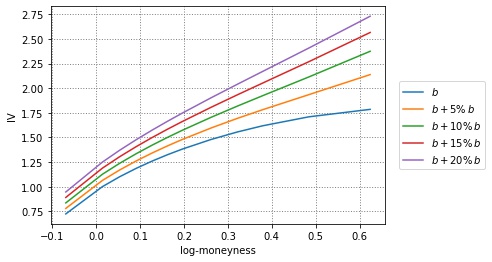} 
     \includegraphics[width=5cm]{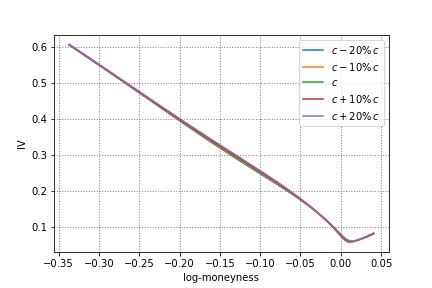} \hspace{-0.75cm}
     \includegraphics[width=5cm]{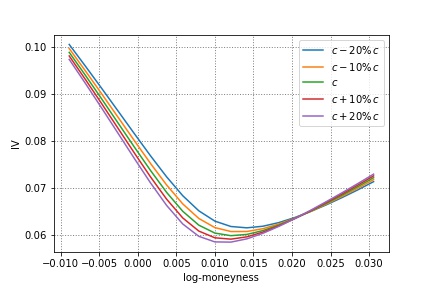} \hspace{-0.75cm}
      \includegraphics[width=5.5cm]{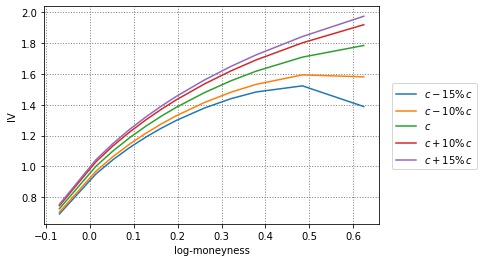} 
       \includegraphics[width=5cm]{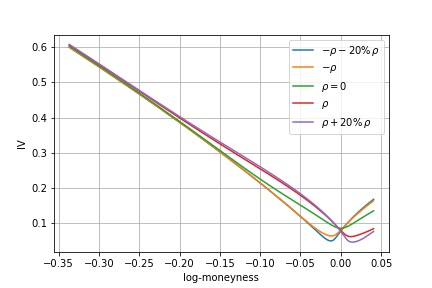} \hspace{-0.75cm}
     \includegraphics[width=5cm]{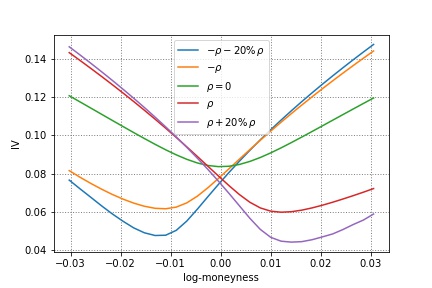} \hspace{-0.75cm}
      \includegraphics[width=5.5cm]{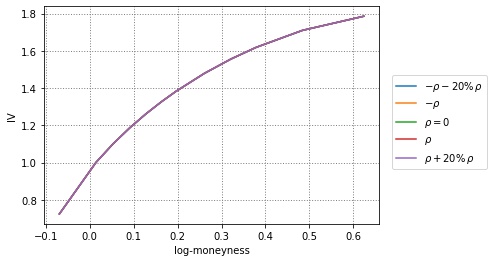} 
        \includegraphics[width=5cm]{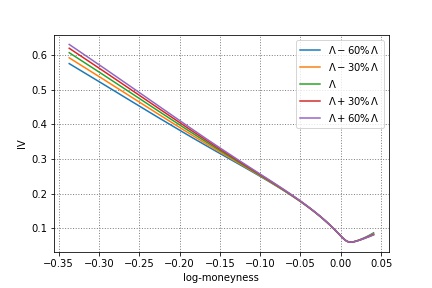} \hspace{-0.05cm}
     \includegraphics[width=5cm]{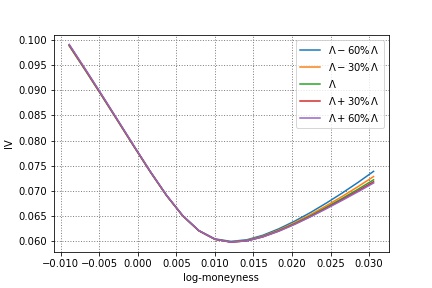}  \hspace{-0.05cm}
      \includegraphics[width=5.5cm]{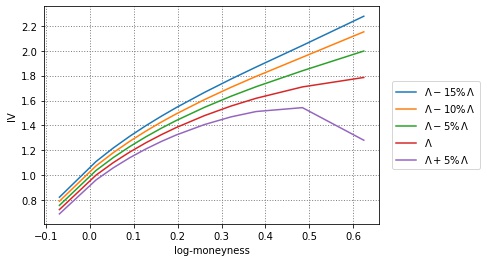} 
      \vspace{-0.1cm} 
    \caption{Sensitivity of implied volatility for SPX (left, center) and VIX (right) options for the shortest maturity
     with respect to: the volatility of volatility $c$ (first line), the correlation $\rho$ (second line),
     the jump-leverage $\Lambda$ (third line), and the mean reversion speed parameter $b$ (fourth line).}
  \label{fig:sensitivity-others}
\end{figure}
\end{center} 
\vspace{-2ex}
 \begin{center}
\begin{figure}[h!]
  %  \centering
   \vspace{-0.4cm} 
    \includegraphics[width=3.7cm]{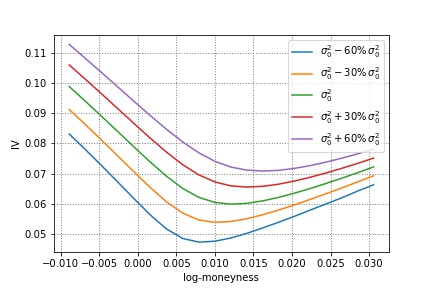} \hspace{-0.5cm}
       \includegraphics[width=3.7cm]{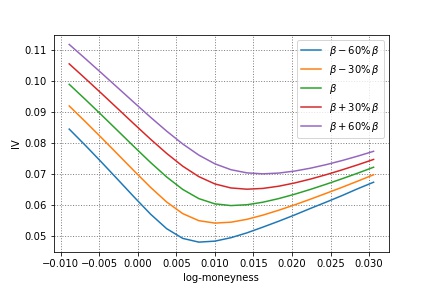} \hspace{-0.5cm}
      \includegraphics[width=4.2cm]{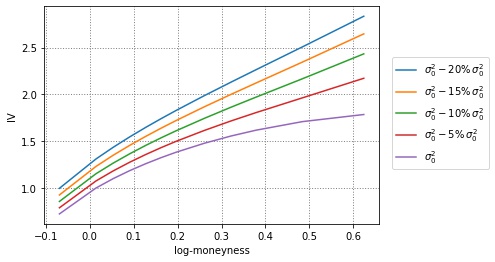}  \hspace{-0.2cm}
      \includegraphics[width=4.2cm]{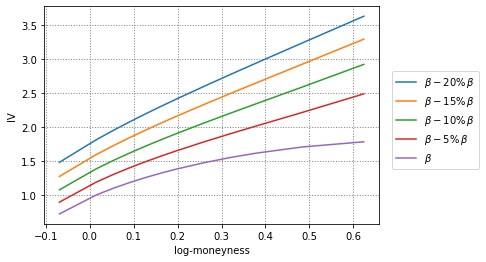} 
      \includegraphics[width=3.7cm]{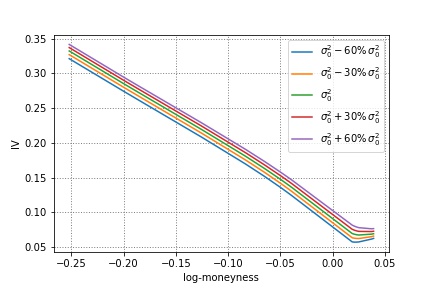} \hspace{-0.3cm}
     \includegraphics[width=3.7cm]{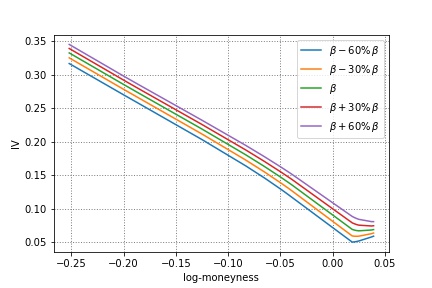}\hspace{-0.15cm}
      \includegraphics[width=4.2cm]{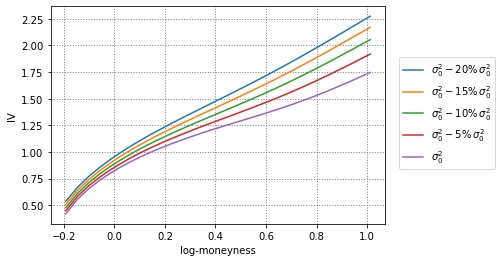}  \hspace{+0.2cm}
      \includegraphics[width=4.2cm]{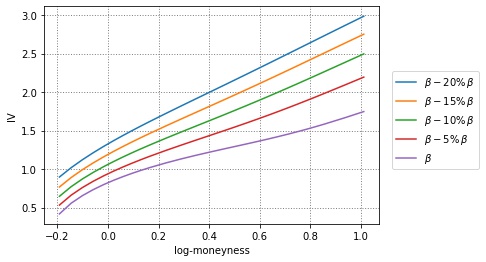} 
      \vspace{-0.1cm} 
    \caption{Sensitivity of implied volatility for SPX (left, left-center)  and VIX (right-center, right) options for the shortest (first line)
    and longest maturity (second line)
     with respect to the initial spot variance curve, i.e. intercept $\sigma_0^2$, and proportional coefficient $\beta$.}
  \label{fig:sensitivity-spot-curve}
\end{figure}
\end{center} 
%%%%%  %%%%%%%
\hspace{0pt}
\vfill

\newpage

	%\nocite{*}

\end{document}